\documentclass[sigconf]{aamas}  

\AtBeginDocument{%
  \providecommand\BibTeX{{%
    \normalfont B\kern-0.5em{\scshape i\kern-0.25em b}\kern-0.8em\TeX}}}
    
\usepackage{graphicx}
\usepackage{subfig}
\usepackage{natbib}
\usepackage{multirow}
\usepackage{booktabs}    
\usepackage{flushend} 

\setcopyright{ifaamas}  
\copyrightyear{2020} 
\acmYear{2020} 
\acmDOI{} 
\acmPrice{} 
\acmISBN{} 
\acmConference[AAMAS'20]{Proc.\@ of the 19th International Conference on Autonomous Agents and Multiagent Systems (AAMAS 2020)}{May 9--13, 2020}{Auckland, New Zealand}{B.~An, N.~Yorke-Smith, A.~El~Fallah~Seghrouchni, G.~Sukthankar (eds.)}  

\begin{document}

\title{RMB-DPOP: Refining MB-DPOP by Reducing Redundant Inferences}  


\author{Ziyu Chen}
\affiliation{
  \institution{College of Computer Science,
  Chongqing University}
  \city{Chongqing} 
  \state{China} 
  \postcode{400044}
}
\email{chenziyu@cqu.edu.cn}

\author{Wenxin Zhang}
\affiliation{
  \institution{College of Computer Science,
  Chongqing University}
  \city{Chongqing} 
  \state{China} 
  \postcode{400044}
}
\email{wenxinzhang18@163.com}

\author{Yanchen Deng}
\authornote{Corresponding author}
\affiliation{
  \institution{College of Computer Science,
  Chongqing University}
  \city{Chongqing} 
  \state{China} 
  \postcode{400044}
}
\email{dyc941126@126.com}

\author{Dingding Chen}
\affiliation{
  \institution{College of Computer Science,
  Chongqing University}
  \city{Chongqing} 
  \state{China} 
  \postcode{400044}
}
\email{dingding@cqu.edu.cn}

\author{Qing Li}
\affiliation{
  \institution{College of Electrical Engineering,
  Chongqing University}
  \city{Chongqing} 
  \state{China} 
  \postcode{400044}
}
\email{qiangli.ac@gmail.com}


\begin{abstract}
  MB-DPOP is an important complete algorithm for solving Distributed Constraint Optimization Problems (DCOPs) by exploiting a cycle-cut idea to implement memory-bounded inference. However, each cluster root in the algorithm is responsible for enumerating all the instantiations of its cycle-cut nodes, which would cause redundant inferences when its branches do not have the same cycle-cut nodes. Additionally, a large number of cycle-cut nodes and the iterative nature of MB-DPOP further exacerbate the pathology. As a result, MB-DPOP could suffer from huge coordination overheads and cannot scale up well. Therefore, we present RMB-DPOP which incorporates several novel mechanisms to reduce redundant inferences and improve the scalability of MB-DPOP. First, using the independence among the cycle-cut nodes in different branches, we distribute the enumeration of instantiations into different branches whereby the number of nonconcurrent instantiations reduces significantly and each branch can perform memory bounded inference asynchronously. Then, taking the topology into the consideration, we propose an iterative allocation mechanism to choose the cycle-cut nodes that cover a maximum of active nodes in a cluster and break ties according to their relative positions in a pseudo-tree. Finally, a caching mechanism is proposed to further reduce unnecessary inferences when the historical results are compatible with the current instantiations. We theoretically show that with the same number of cycle-cut nodes RMB-DPOP requires as many messages as MB-DPOP in the worst case and the experimental results show our superiorities over the state-of-the-art.
\end{abstract}


\maketitle


\section{Introduction}
Distributed constraint optimization problems (DCOPs) are a fundamental framework for coordinated and cooperative multi-agent systems. DCOPs have been successfully applied to model many real-world problems including sensor networks \cite{dsa2005}, meeting scheduling \cite{taskmodeling2007}, smart grid \cite{grid2017}, etc. 

Incomplete algorithms for DCOPs \cite{mgm2004,dsa2005,k-optimal2007,maxsum2008,dgibbs2019,chen2018class} aim to find a good solution in an acceptable overhead, while complete algorithms focus on finding the optimal one by employing either search or inference to systematically explore the entire solution space. SBB \cite{sbb1997}, AFB \cite{afb2009}, PT-FB \cite{ptfb2017}, ADOPT \cite{adopt2005} and its variants \cite{adopt_trading2009,adopt_caching2009,bnbadopt2010,bnbadopt_generalizing2011,bnbadopt_remove2012} are typical search-based algorithms that employ distributed backtrack search to exhaust the search space. However, these algorithms incur a prohibitively large number of messages and can only solve the problems with a handful of variables.

On the other hand, inference-based complete algorithms like DPOP \cite{dpop2005} perform dynamic-programming on a pseudo-tree and only require a linear number of messages. However, the memory consumption in DPOP is exponential to the induced width \cite{cyclecut2003}, which makes it not applicable for the memory-limited scenarios where the optimal solution is desired \cite{heterogeneous2018,vessel2016}. Therefore, a number of algorithms \cite{adpop2005,lsdpop2007,mbdpop2007,odpop2006,funcFilter2010} were proposed to trade either solution quality or message number for smaller memory consumption. Among these algorithms, MB-DPOP \cite{mbdpop2007} iteratively performs memory-bounded utility propagation to guarantee the optimality. Specifically, given the dimension limit $k$, the algorithm first identifies high-width areas (clusters) and cycle-cut nodes \cite{cyclecut2003} to make the maximal dimension of the utility tables propagated within clusters no greater than $k$. For each cluster, the cluster root is responsible for iteratively enumerating all the instantiations of its cycle-cut nodes, and nodes in the cluster perform memory-bounded inferences by conditioning utility tables on these instantiations. Once instantiations are exhausted, the cluster root propagates the resulted utility table to its parent.

However, a key limitation in MB-DPOP is the inability of exploiting the structure of a problem. As a result, MB-DPOP suffers from a severe redundancy in memory-bounded inference. First, since each cluster root enumerates for all its cycle-cut nodes without considering the independence of cycle-cut nodes in different branches, each branch in a cluster would have to perform redundant inferences when there are cycle-cut nodes which have nothing to do with the branch. Also, these nonconcurrent instantiations severely degenerate the parallelism among branches. Second, agents in a cluster use heuristics to determine cycle-cut nodes locally, which would results in a large number of cycle-cut nodes. Finally, MB-DPOP ignores the validity of the inference results and a branch has to perform inference even if the previous results are compatible with the current instantiations.

In this paper, we aim to improve the scalability of MB-DPOP by exploiting the structure of a problem. More specifically, our contributions are listed as follows.
\begin{itemize}
	\item By using the independence among the cycle-cut nodes in different branches, we propose a distributed enumeration mechanism where a cluster root only enumerates for cycle-cut nodes in its separators and these instantiations are augmented with branch-specific cycle-cut nodes dynamically along the propagation. Accordingly, each branch can perform memory-bounded inferences asynchronously and the number of nonconcurrent instantiations can be reduced significantly.
	\item We propose an iterative selection mechanism to determine cycle-cut nodes by taking both their effectiveness and positions in a pseudo tree into consideration. Concretely, rather than choosing highest/lowest separators as cycle-cut nodes, we tend to choose the nodes that cover a maximum of active nodes in a cluster and break ties according to their relative positions. Moreover, we propose a caching mechanism to exploit the historical inference results which are compatible with the current instantiations to further avoid unnecessary utility propagation.
	\item We theoretically show that the message number of our algorithm is no more than the one in MB-DPOP. Our empirical evaluation confirms the superiority of our algorithm over the state-of-the-art on various benchmarks.
\end{itemize}

The rest of this paper is organized as follows. Section 2 gives the background including DCOPs, pseudo tree, DPOP and MB-DPOP. In section 3, we give the motivation and describe the details of our proposed algorithm. And the experiments are shown in Section 4. Section 5 concludes this paper and gives future research work. 
        
\section{Background}
In this section, we introduce the preliminaries including DCOPs, pseudo tree, DPOP and MB-DPOP.

\subsection{Distributed Constraint Optimization Problems}
A distributed constraint optimization problem \cite{dpop2005} can be represented by a tuple $\langle A,X,D,F\rangle$ where
\begin{itemize}
	\item $A=\{a_1,\dots, a_n\}$ is a set of agents
	\item $X=\{x_1,\dots,x_m\}$ is a set of variables
	\item $D=\{D_1,\dots, D_m\}$ is a set of domains that are finite and discrete, each variable $x_i$ taking a value assignment from $D_i$
	\item $F=\{f_1,\dots,f_q\}$ is a set of constraint functions, each function $f_i:D_{i1}\times\dots\times D_{ik}\rightarrow\mathbb{R}_{\ge 0}$ denoting the non-negative cost for each assignment combination of $x_{i1},\dots,x_{ik}$.
\end{itemize}

For the sake of simplicity, we assume that each agent controls a variable and all constraint functions are binary (i.e., $f_{ij}: D_i\times D_j\rightarrow \mathbb{R}_{\ge 0}$). Here, the term "agent" and "variable" can be used interchangeably. A solution to a DCOP is an assignment including all variables that makes the minimum cost. That is
\begin{displaymath}
    X^*=\mathop{\arg\min}_{d_i\in D_i, d_j\in D_j}\sum_{f_{ij}\in F}f_{ij}(x_i=d_i,x_j=d_j)
\end{displaymath}

A DCOP can be visualized by a constraint graph presented as Fig. 1, where the nodes represent the agents and the edges represent the constraints, respectively.
\begin{figure}
	\centering
	\includegraphics[scale=0.4]{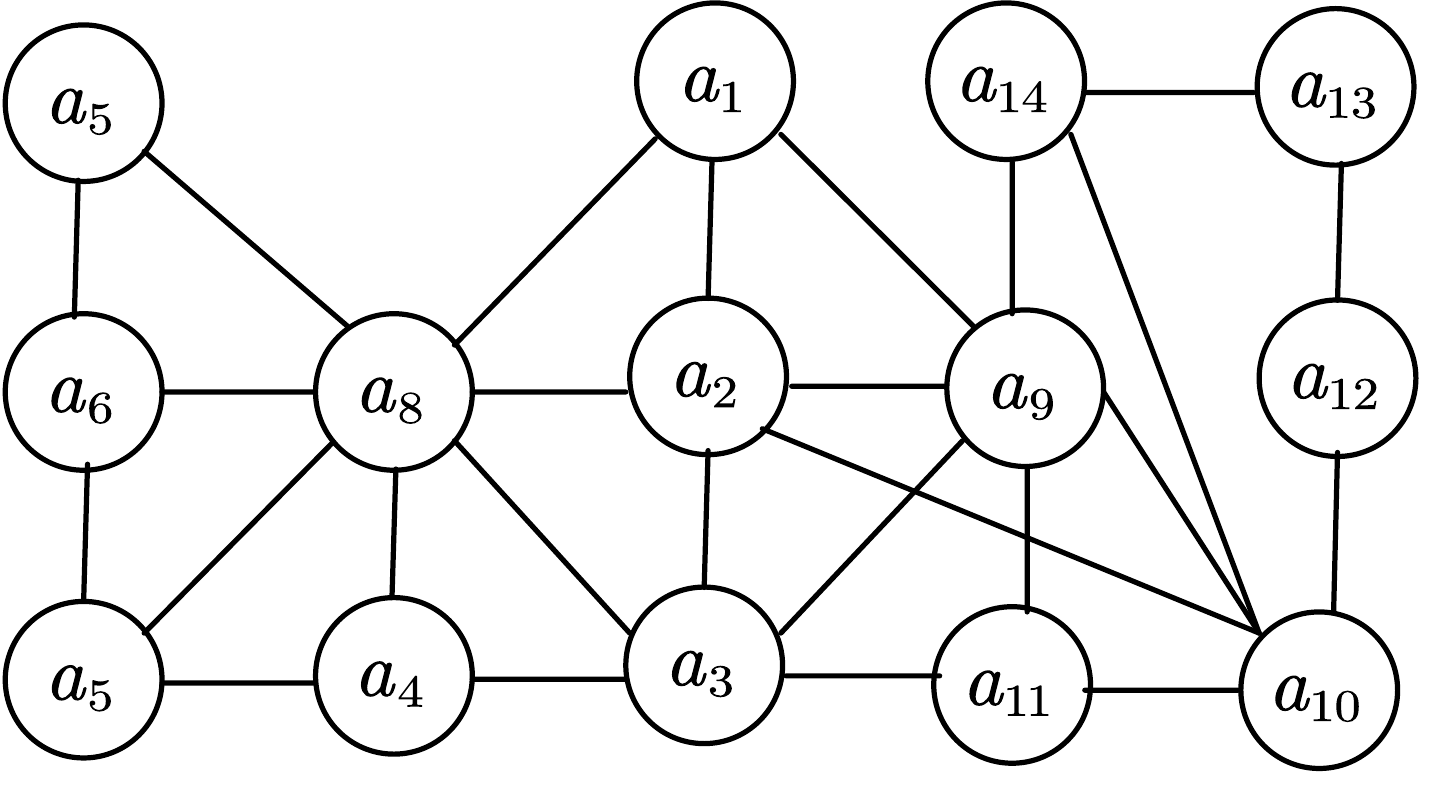}
	\caption{The constraint graph of a DCOP}
\end{figure}

\begin{figure}
	\centering
	\includegraphics[scale=0.45]{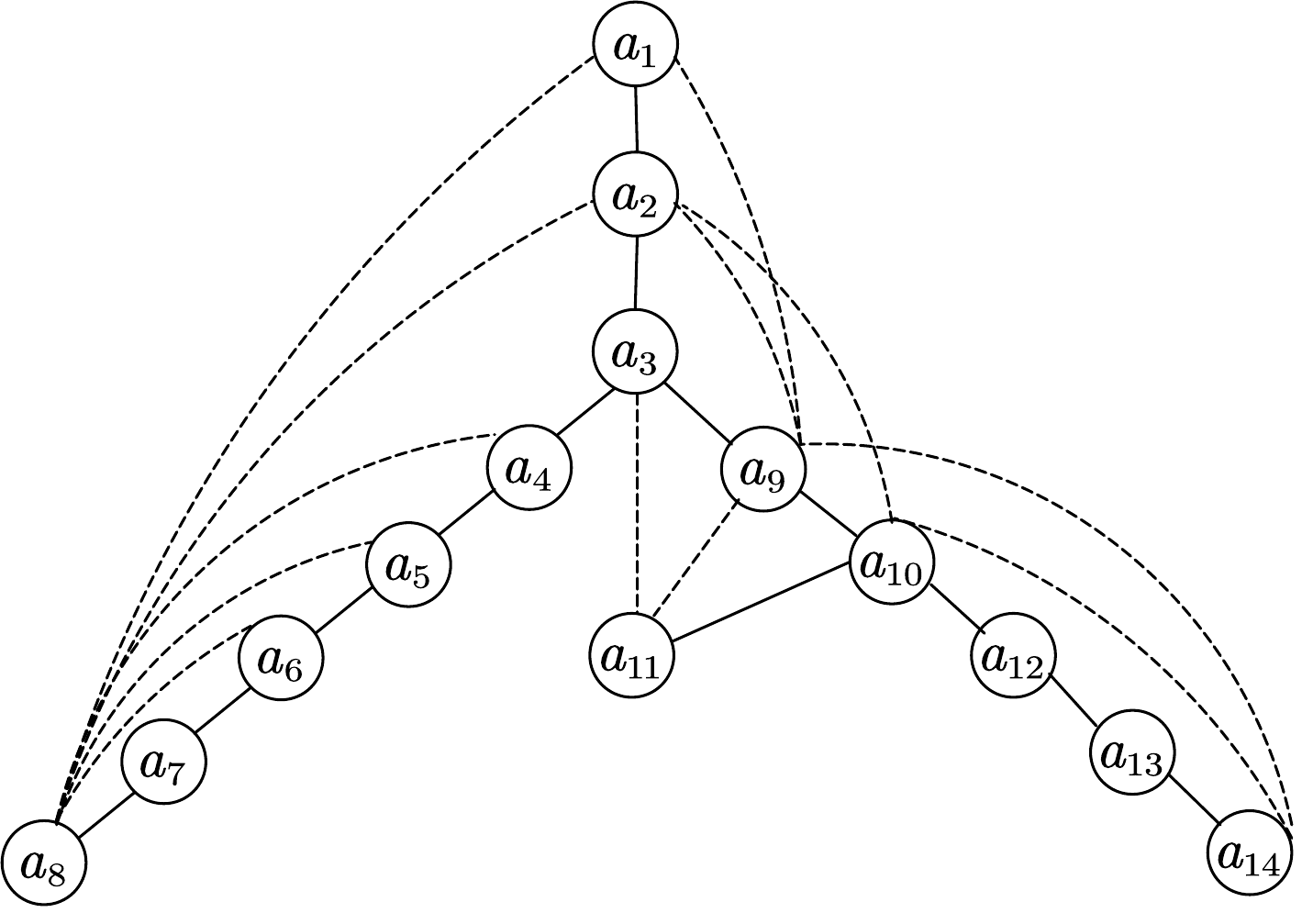}
	\caption{A pseudo tree derived from Figure 1}
\end{figure}

\subsection{Pseudo Tree}
A pseudo tree \cite{freuderQ85} is a partial ordering among agents and can be generated by depth-first search (DFS) traversal to a constraint graph, where different branches are independent from each other. Given a constraint graph and its spanning tree, the edges in the spanning tree are tree edges and the other edges are pseudo edges (i.e., non-tree edges). According to the relative positions in a pseudo tree, the neighbors of an agent $a_i$ connected by tree edges are categorized into parent $P(a_i)$ and children $C(a_i)$, while the ones connected by pseudo edges are denoted as pseudo parents $PP(a_i)$ and pseudo children $PC(a_i)$. For its parent and pseudo parents, we denote them as $AP(a_i)=\{P(a_i)\}\cup PP(a_i)$. We also denote its descendants as $Desc(a_i)$. Finally, the separators $Sep(a_i)$ \cite{odpop2006} of $a_i$ refer to the ancestors which are constrained with $a_i$ or its descendants. Fig. 2 gives a pseudo tree derived from the constraint graph in Fig. 1, where tree edges and pseudo edges are denoted by solid and dotted lines, respectively.

\subsection{DPOP}
DPOP \cite{dpop2005} is an inference-based complete algorithm for DCOPs, which implements the bucket elimination in a distributed manner \cite{dechter1999bucket}. It performs two phases of propagation on a pseudo tree via tree edges: a UTIL propagation phase to eliminate dimensions from the bottom up, and a VALUE propagation phase to assign the optimal value for each variable vice versa along the pseudo tree. More specifically, in a UTIL propagation phase, an agent $a_i$ collects the utility tables from its children and joins them with the local utility table, then computes the optimal utilities for all possible assignments of $Sep(a_i)$ to eliminate its dimension from the joint utility table. Then, it sends the projected utility table to its parent. In a VALUE propagation phase, once $a_i$ receives the assignments from its parent, it plugs them into the joint utility table obtained in the UTIL propagation phase to get the optimal assignment, and sends the joint assignment to its children. Although DPOP only requires a linear number of messages to solve a DCOP, its space complexity is exponential in the induced width of the pseudo tree. 

\subsection{MB-DPOP}
MB-DPOP \cite{mbdpop2007} attempts to improve the scalability of DPOP by trading the message number for smaller memory consumption. Given the dimension limit $k$, MB-DPOP starts with a labeling phase to identify the areas with the induced width \cite{cyclecut2003} higher than $k$ (i.e., clusters) and the corresponding cycle-cut (CC) nodes. Each cluster is bounded at the top by the lowest node in the tree that has separators of size $k$ or less, and such node is called the cluster root (CR) node.  
For each clusters, CC nodes are determined such that the cluster has the width no greater than $k$ once they are removed. 
	
In more detail, the CC nodes are selected and then aggregated in a bottom-up fashion. 
That is, given the lists of CC nodes selected by its children, $a_i$ first determines whether its width exceeds $k$ if $|Sep(a_i)|>k$. If it is the case, $a_i$ needs to choose additional CC nodes to enforce the memory limit $k$ by a heuristic function. Then, $a_i$ propagates all the CC nodes $CClist_i$ to its parent $P(a_i)$. Otherwise, if $|Sep(a_i)|\le k$ and the lists received from children are all empty, $a_i$ labels self as a normal node and propagates the utility as in DPOP. 

During the UTIL propagation phase normal nodes (i.e., the nodes whose width is no greater than $k$) perform canonical utility propagation while the other nodes in each cluster perform memory-bounded inferences. Specifically, each cluster root (CR) enumerates instantiations for its CC nodes and propagates them iteratively to the other nodes in the cluster, and these nodes perform memory-bounded inferences by conditioning the utility tables on the received instantiation. The cluster root eliminates its dimension and propagates the utility table to its parent after exhausting all the combinations. Finally, a VALUE propagation phase starts. Different from DPOP which only requires a round of value propagation, MB-DPOP requires additional utility propagation to re-drive the utilities corresponding to the assignments of CC nodes to get the optimal values, since non-CC nodes in a cluster only cache the utility table for the latest instantiation of CC nodes. 

\section{Proposed Method}
In this section, we present our proposed RMB-DPOP. We begin with a motivation, and then present the details and the theoretical claim of our algorithm, respectively.
\subsection{Motivation}0 
As stated earlier, MB-DPOP suffers from plenty of redundancies in memory-bounded inference due to the inability of exploiting a problem structure in both instantiation enumeration and the selection of CC nodes. Consider the problem in Fig. 2, where $a_3$ is the only cluster root with the dimension limit $k=2$. Since each agent in the cluster selects its CC nodes only with the local knowledge, MB-DPOP would select a large number of CC nodes and significantly increase the number of instantiations. In fact, if we choose CC nodes with the highest level, the CC nodes of $a_3$ are $CClist_{3} = \{a_1, a_2, a_3, a_4, a_5, a_9\}$. It would be worse when using the lowest heuristic which results in 9 CC nodes in this case. Alternatively, instead of choosing both $a_3$ and $a_9$, we could only choose $a_9$ and still guarantee the memory budget. Besides, the cluster root $a_3$ has to enumerate all instantiations of $CClist_{3}$, which results in a large number of nonconcurrent instantiations and redundant inferences. In fact, we could exploit the independence between branch $a_4$ and branch $a_9$ by generating instantiations that only contains the common CC nodes (i.e., $a_1, a_2$ and $a_3$). In this way, branch $a_4$ and branch $a_9$ can operate asynchronously and the number of non-concurrent instantiations is significantly reduced. In addition, all the bounded inference results are disposable in MB-DPOP, which also leads to redundant inferences. In fact, some inference results received from children in the previous iterations are compatible with the current instantiation, since each branch performs memory-bounded inference by conditioning only on a subset of all cycle-cut nodes of a cluster. Thus, it is unnecessary to perform a memory-bounded inference when the assignments of corresponding CC nodes do not change.

Therefore, to take the structure of a problem into consideration, we propose a novel algorithm named RMB-DPOP which incorporates a distributed enumeration mechanism to reduce the nonconcurrent instantiations, an iterative selection mechanism to reduce the number of CC nodes and a caching mechanism to avoid unnecessary inferences. Algorithm 1 \footnote{We omit the details of the value propagation phase due to its similarity to the one in MB-DPOP. The source code is available in https://github.com/czy920/RMB-DPOP. } presents the sketch of RMB-DPOP.

\subsection{Distributed Enumeration Mechanism}
\begin{figure}
	\centering
	\includegraphics[scale=1.0]{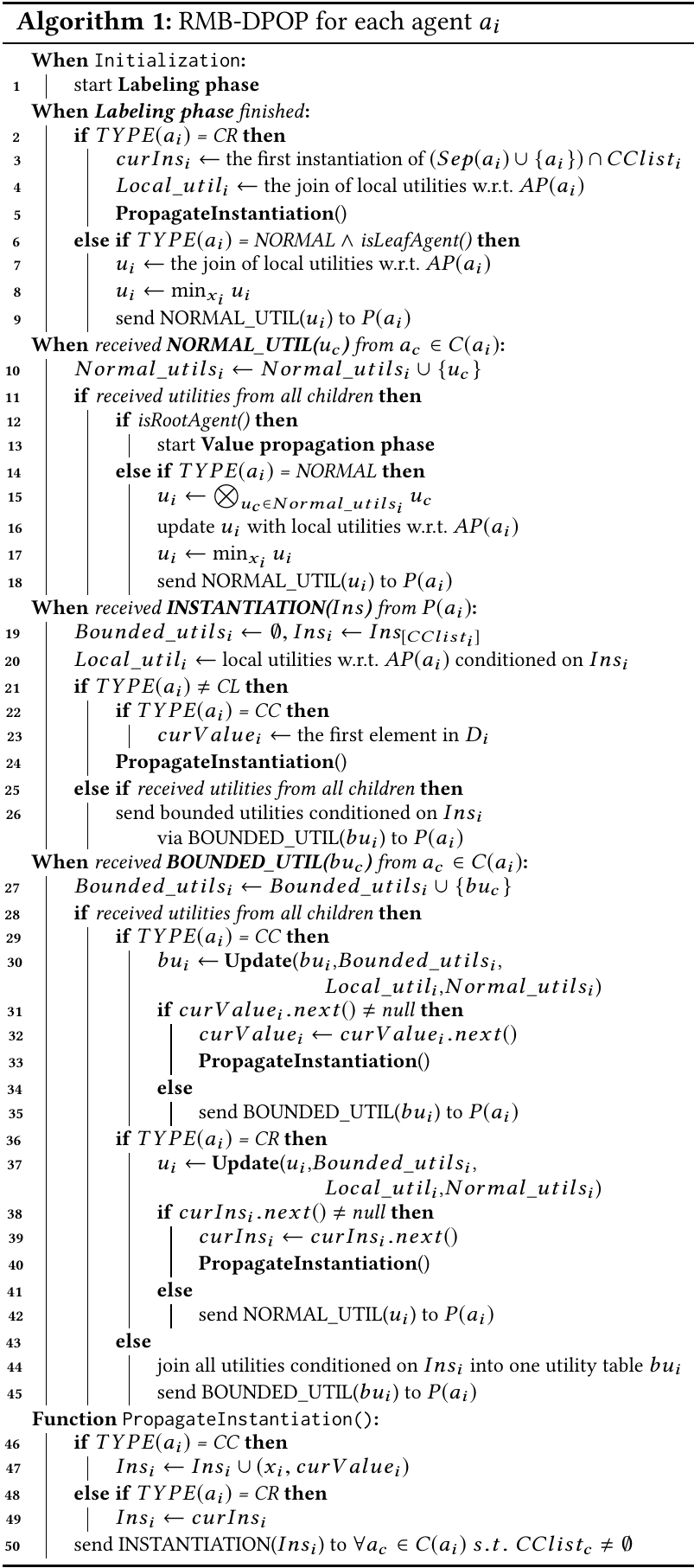}
\end{figure}
Distributed enumeration mechanism (DEM) is adopted in each cluster to perform asynchronous memory-bounded inference by factorizing the instantiations. More specifically, since each branch in a pseudo tree is independent, each CC node inside a cluster is only related to a subproblem. Hence, instead of enumerating all the instantiations of CC nodes by a cluster root, we only generate instantiations for the CC nodes in the separators of the cluster root and dynamically augment these instantiations with branch-specific CC nodes. In the following, we present the details of the mechanism.

When the Labeling phase finishes, a CR node $a_j$ starts the iterative memory-bounded UTIL propagation by instantiating the nodes in $(Sep(a_j)\cup \{a_j\}) \cap CClist_j$ (line 2-5), where the $CClist_j$ is a list of the CC nodes corresponding to the branch of $a_j$. When a CC node $a_i$ receives an instantiation $Ins_i$ from its parent, it augments $Ins_i$ by the first assignment from its domain and propagates the extended instantiation to its children in the cluster (line 22-24, 46-50). Once $a_i$ receives all the utilities from its children, it updates the cache and replaces self assignment with the next value in its domain, and then propagates the new instantiation (line 29-33). Until getting a complete bounded inference result corresponding to $Ins_i$ by a traversal of its domain, $a_i$ sends the result to its parent via a BOUNDED\_UTIL message (line 35).

Next, we theoretically show its superiority over MB-DPOP in terms of the message number. Let us first introduce two notations. For a cluster root $a_j$, we denote $CClist_j^{out} = (Sep(a_j)\cup \{a_j\}) \cap CClist_j$ as the set of CC nodes enumerated by $a_j$, and the remaining CC nodes as $CClist_j^{in} = CClist_j\backslash CClist_j^{out}$. 

\begin{lemma}
	For an agent $a_i$ in a cluster where $a_j$ is the CR node, the number of instantiations it receives is exponential in the size of $(CClist_i \cap Sep(a_i)) \cup CClist_j^{out}$.
\end{lemma}
\begin{proof}
According to line 2-5, the nonconcurrent instantiations sent from $a_j$ is exponential in $|CClist_j^{out}|$. Besides, each nonconcurrent instantiation is augmented by the CC nodes along the path from $a_j$ to $a_i$ (line 23, 32). Therefore, the number of instantiations $a_i$ receives is exponential in $|(CClist_i \cap Sep(a_i)) \cup CClist_j^{out}|$.
\end{proof}

\begin{theorem}
	Given the same CC lists of each node in each cluster, the maximal message number of RMB-DPOP is no more than the one in MB-DPOP.
\end{theorem}

\begin{proof}
	It is enough to show the theorem by analyzing the total number of instantiations received by each agent $a_i$ in a cluster, since $a_i$ must respond with a bounded utility table to its parent after receiving an instantiation. Without loss of generality, we assume that each variable has a domain with the same size $d$. In MB-DPOP each agent in a cluster will receive $d^{|CClist_j|}$ instantiations. Whereas from Lemma 1, we have the number of instantiations sent to $a_i$ as $d^{|(CClist_i \cap Sep(a_i)) \cup CClist_j^{out}|}$, and
	$$
	\small
	\begin{aligned}
	|CClist_j|
	&= |CClist_j^{out} \cup CClist_j^{in}|\\
	&= |CClist_j^{out}|+|CClist_j^{in}|\\
	&\geq |CClist_j^{out}|+|CClist_i\cap CClist_j^{in} \cap Sep(a_i)|\\
	&= |CClist_j^{out}\cup (CClist_i\cap CClist_j^{in} \cap Sep(a_i))|\\
	&= |(CClist_i \cap Sep(a_i)) \cup CClist_j^{out}| 
	\end{aligned}
	$$
	Consequently, RMB-DPOP propagates a smaller number of instantiations than MB-DPOP. And only when the cluster does not have the CC nodes inside the cluster (i.e., $CClist_j^{in} = \emptyset$), the instantiations for each agent in RMB-DPOP are equivalent to those in MB-DPOP. Thus, the theorem holds.
\end{proof}

\begin{figure}
	\centering
	\includegraphics[scale=1.0]{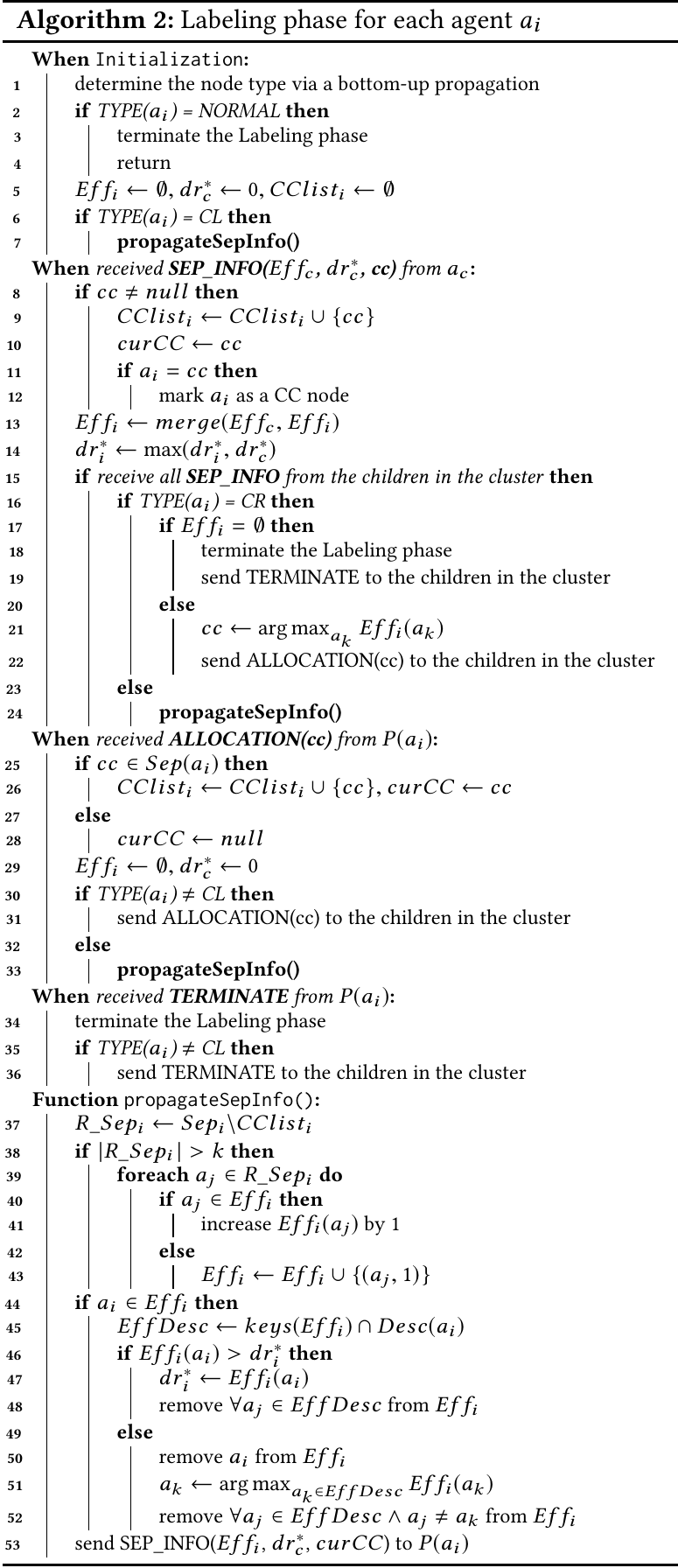}
\end{figure}

\subsection{Iterative Selection Mechanism}
Instead of selecting CC nodes based on the local knowledge in MB-DPOP which would result in a large number of CC nodes, we choose CC nodes by taking their effectiveness and their relative positions into consideration through an iterative selection mechanism (ISM). Specifically, in a cluster we measure the effectiveness of a node by the number of active nodes it covers. Here, an active node is the one whose width is still greater than $k$ given the selected CC nodes. Besides, to facilitate DEM, we tend to select nodes in different branches of a cluster. Therefore, we propose to break ties among the nodes with the same effectiveness by their positions in a pseudo tree. Algorithm 2 gives the sketch of Labeling phase. 

In more detail, a CC node is selected through two phases of message-passing. In the first phase, the effectiveness of each CC node candidate is aggregated in a bottom-up fashion via SEP\_INFO messages. Specifically, each agent $a _i$ maintains a data structure $Eff_i$ to record the effectiveness of candidates. When receiving a SEP\_INFO message from a child $a_c$, it updates $Eff_i$ by $Eff_c$ according to

\begin{displaymath}
	\small
	Eff_i\left( a_k \right) =\begin{cases}
		Eff_c\left( a_k \right) ,a_k\notin keys\left( Eff_i \right)\\
		Eff_i\left( a_k \right) + Eff_c\left( a_k \right) ,otherwise\\
	\end{cases},\forall a_k\in keys\left( Eff_c \right) 
\end{displaymath}

If $a_i$ is an active node (i.e., satisfying line 38), for each CC node candidate $a_k\in R\_Sep_i$ it increases the effectiveness $Eff_i(a_k)$ by 1 (line 39-43). Then $a_i$ removes all the CC candidates that have a suboptimal effectiveness in its descendants from $Eff_i$ (line 45-52), since they cannot produce the highest effectiveness. The phase ends when the cluster root $a_j$ receives all the SEP\_INFO messages from the children in the cluster.

In the second phase, the cluster root $a_j$ chooses the CC node with the maximal effectiveness (line 21) and propagates it into the cluster via ALLOCATION messages. According to Lemma 3.1 and Theorem 3.2, our algorithm can take the advantage of the CC nodes inside the cluster through the DEM. Therefore, we propose to break ties according to the height of the candidates when choosing a CC node, i.e., we tend to choose the lowest CC node since it is more likely to be inside the cluster. The phase ends after each cluster leaf (CL) starts a new phase of effectiveness propagation (line 33). The Labeling phase terminates when there is no active nodes (i.e., satisfying line 17).

It is worth noting that our selection mechanism only incurs minor messages. Specifically, to determine a CC node in the cluster with CR $a_j$, agents need to propagate bottom-up SEP\_INFO messages and top-down ALLOCATION messages via tree edges, which requires $\mathcal{O} (m)$ messages. Here, $m$ is the total number of nodes in the cluster. Thus, the total messages exchanged in the Labeling phase in a cluster is $\mathcal{O} (|CClist_j|*m)$ and the overall complexity is $\mathcal{O} (N^2)$ where $N$ is the total number of agents.
\begin{table}
	\small
	\caption{The first round of effectiveness aggregation}
	\begin{tabular}{|c|c|c|}
		\hline
		\textbf{step}      & \textbf{message}           & \textbf{SEP\_INFO}                                                        \\ \hline
		\multirow{3}{*}{1} & $a_8\rightarrow a_7$       & $Eff_{8} = \{(a_1, 1), (a_2, 1), (a_4, 1), (a_5, 1), (a_6, 1),(a_7, 1)\}$ \\ \cline{2-3} 
		& $a_{14}\rightarrow a_{13}$ & $Eff_{14} = \{(a_9, 1), (a_{10}, 1), (a_{13}, 1)\}$                       \\ \cline{2-3} 
		& $a_{11}\rightarrow a_{10}$ & $Eff_{11} = \{(a_3, 1), (a_9, 1), (a_{10}, 1)\}$                          \\ \hline
		\multirow{2}{*}{2} & $a_7\rightarrow a_6$       & $Eff_{7} =\{(a_1, 2), (a_2, 2), (a_4, 2), (a_5, 2), (a_6, 2),(a_7,1)\}$   \\ \cline{2-3} 
		& $a_{13}\rightarrow a_{12}$ & $Eff_{13}=\{(a_9, 2), (a_{10}, 2), (a_{12}, 1), (a_{13}, 1)\}$            \\ \hline
		\multirow{2}{*}{3} & $a_6\rightarrow a_5$       & $Eff_{6} =\{(a_1, 3), (a_2, 3), (a_4, 3), (a_5, 3), (a_6, 2)\}$           \\ \cline{2-3} 
		& $a_{12}\rightarrow a_{10}$ & $Eff_{12}=\{(a_9, 3), (a_{10}, 3), (a_{13}, 1)\}$                         \\ \hline
		\multirow{2}{*}{4} & $a_5\rightarrow a_4$       & $Eff_{5} =\{(a_1, 4), (a_2, 4), (a_4, 4), (a_5, 3)\}$                     \\ \cline{2-3} 
		& $a_{10}\rightarrow a_9$    & $Eff_{10}=\{(a_2,1),(a_3, 2),(a_9, 5), (a_{10}, 4)\}$                     \\ \hline
		\multirow{2}{*}{5} & $a_4\rightarrow a_3$       & $Eff_{4} =\{(a_1,5),(a_2,5),(a_3,1),(a_4,4)\}$                            \\ \cline{2-3} 
		& $a_9\rightarrow a_3$       & $Eff_{9}=\{(a_1,1),(a_2,2),(a_3, 3),(a_9, 5)\}$                           \\ \hline
		6                  & N/A                        & $Eff_{3}=\{(a_1,6),(a_2,7),(a_3, 4),(a_4,4),(a_9, 5)\}$                   \\ \hline
		
	\end{tabular}
	\label{effectiveness aggregation}
\end{table}
\begin{figure*}	
	\centering
	\subfloat[Number of Messages]{
		\includegraphics[scale=0.3]{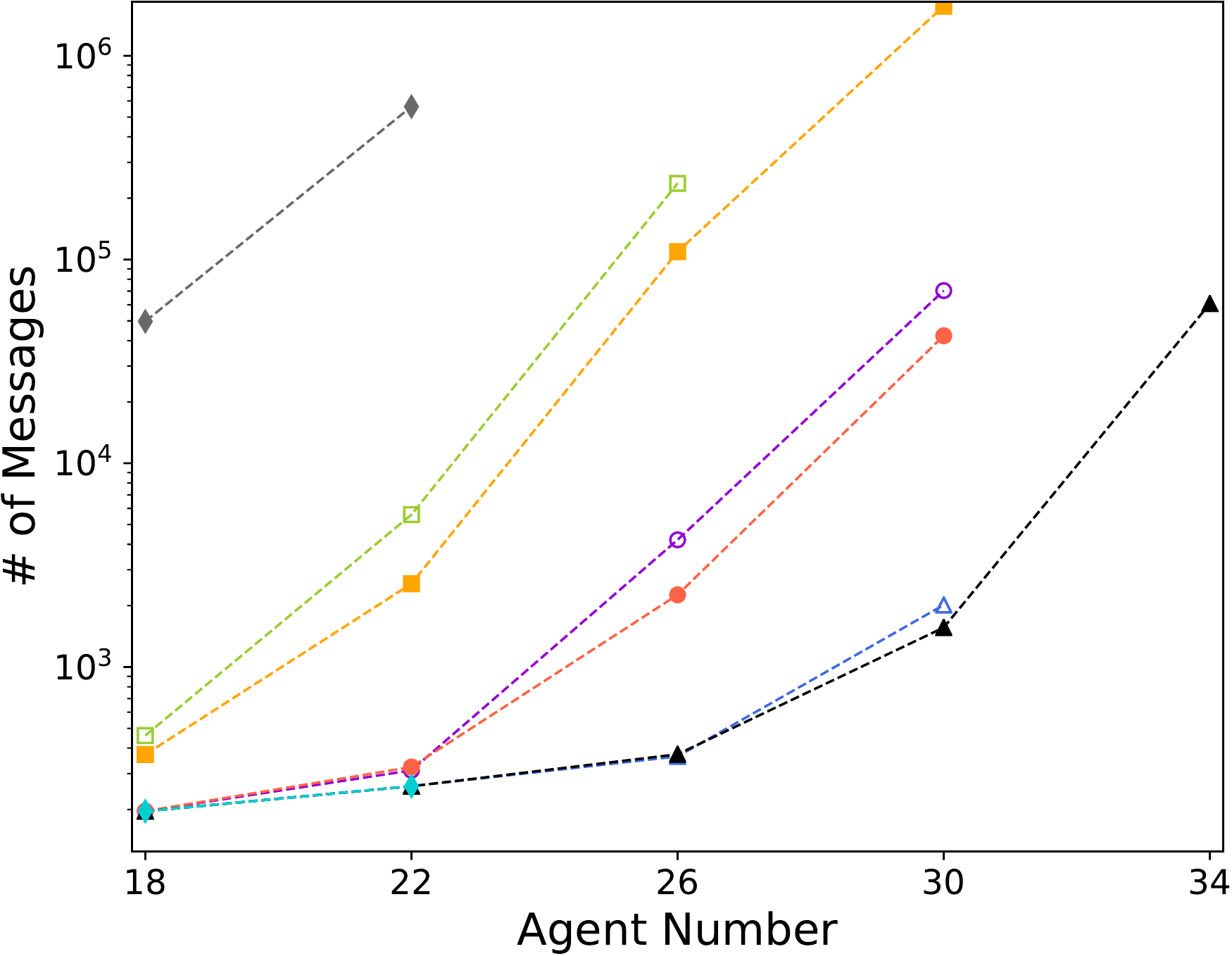} 
	}
	\subfloat[Network Load]{ 
		\includegraphics[scale=0.3]{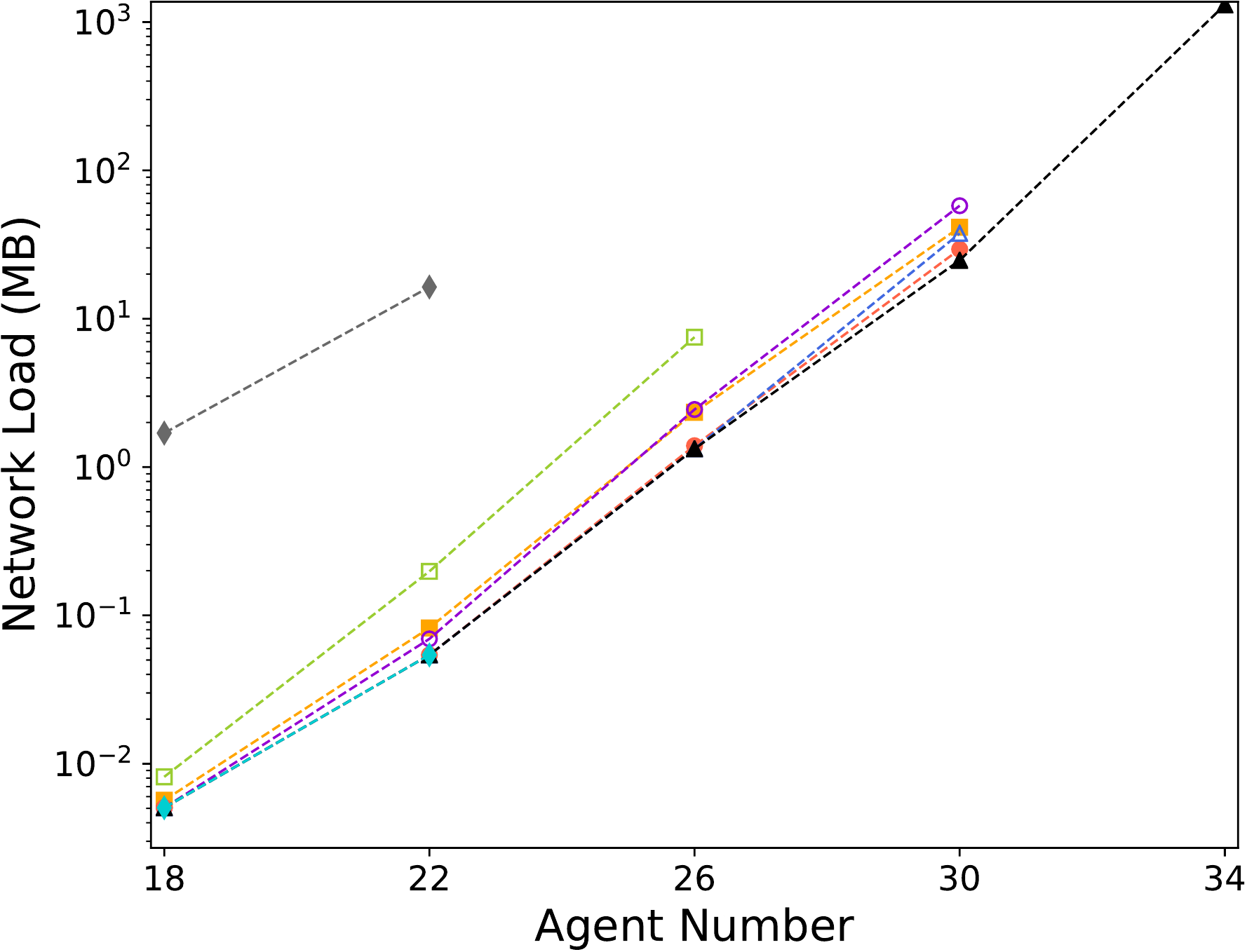} 
	}
	\subfloat[Runtime]{
		\includegraphics[scale=0.3]{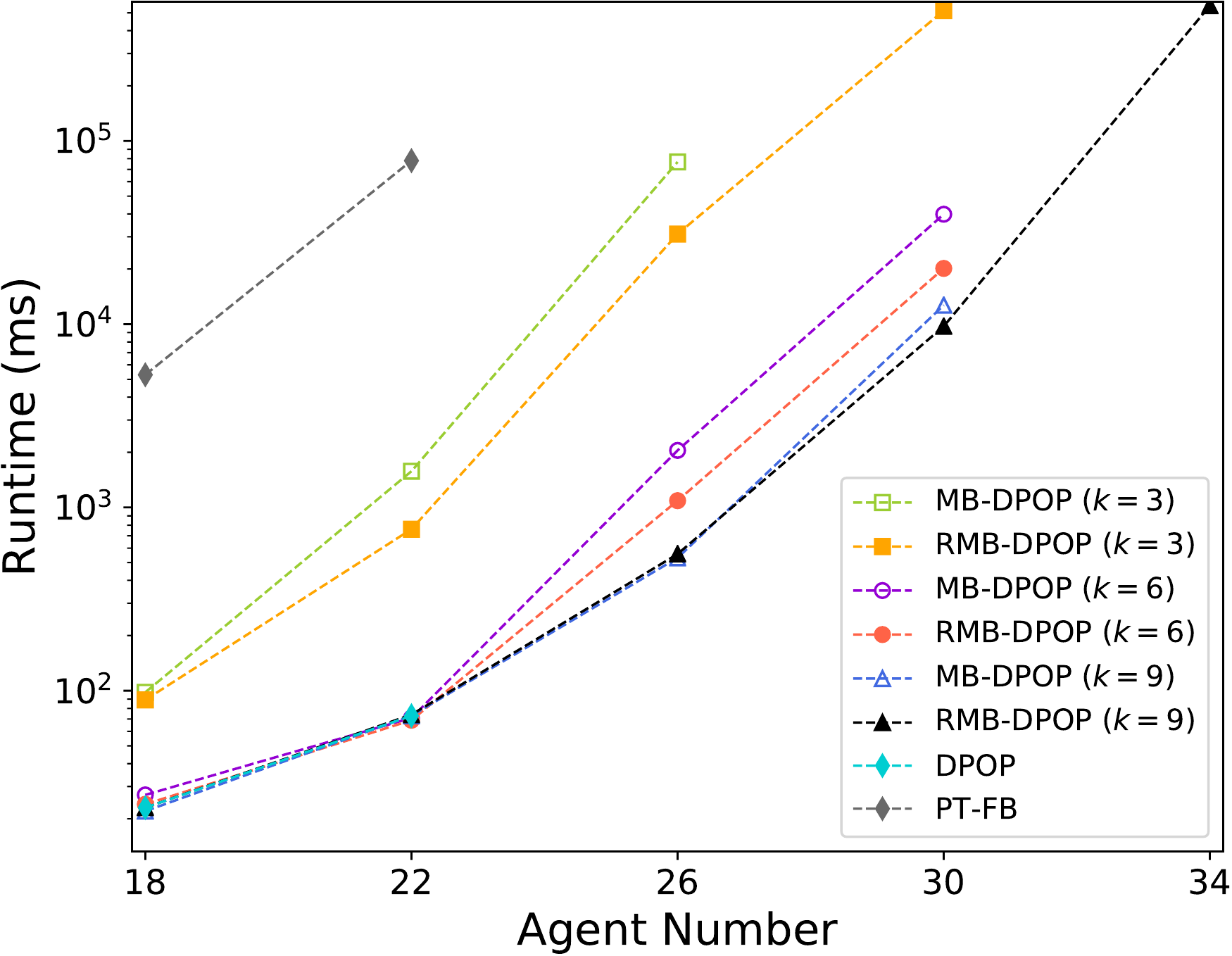} 
	}
	\\
	\caption{Performance comparison under different agent numbers}
	\label{agent}
\end{figure*}

\subsection{Caching Mechanism}
The caching mechanism attempts to reduce unnecessary inferences by exploiting the historical results when they are compatible with the current instantiations. To do this, before $a_i$ propagates an instantiation to a child $a_c$, it projects the instantiation on $CClist_c$ and stores the projected one. When $a_i$ receives a new instantiation, for each child it checks whether the instantiation is compatible with the cached one associated with the child. If it is the case, the results cached in the previous iteration is valid and there is no need to perform a memory-bounded inference. Otherwise, the results from the child is no longer valid and $a_i$ propagates the (augmented) instantiation to the child to initiate a new memory-bounded inference.

\subsection{Execution Example}
For better understanding, we take Fig. 2 as an example to illustrate our algorithm. Assuming the dimension limit $k=2$, there is only a cluster whose CR node is $a_3$. The labeling phase begins with CL nodes $a_8$ and $a_{14}$ which send SEP\_INFO messages to their parents and Table \ref{effectiveness aggregation} presents the trace of effectiveness aggregation in the first round in a chronological order.
\begin{figure*}
	\subfloat[Number of Messages]{
		\includegraphics[scale=0.3]{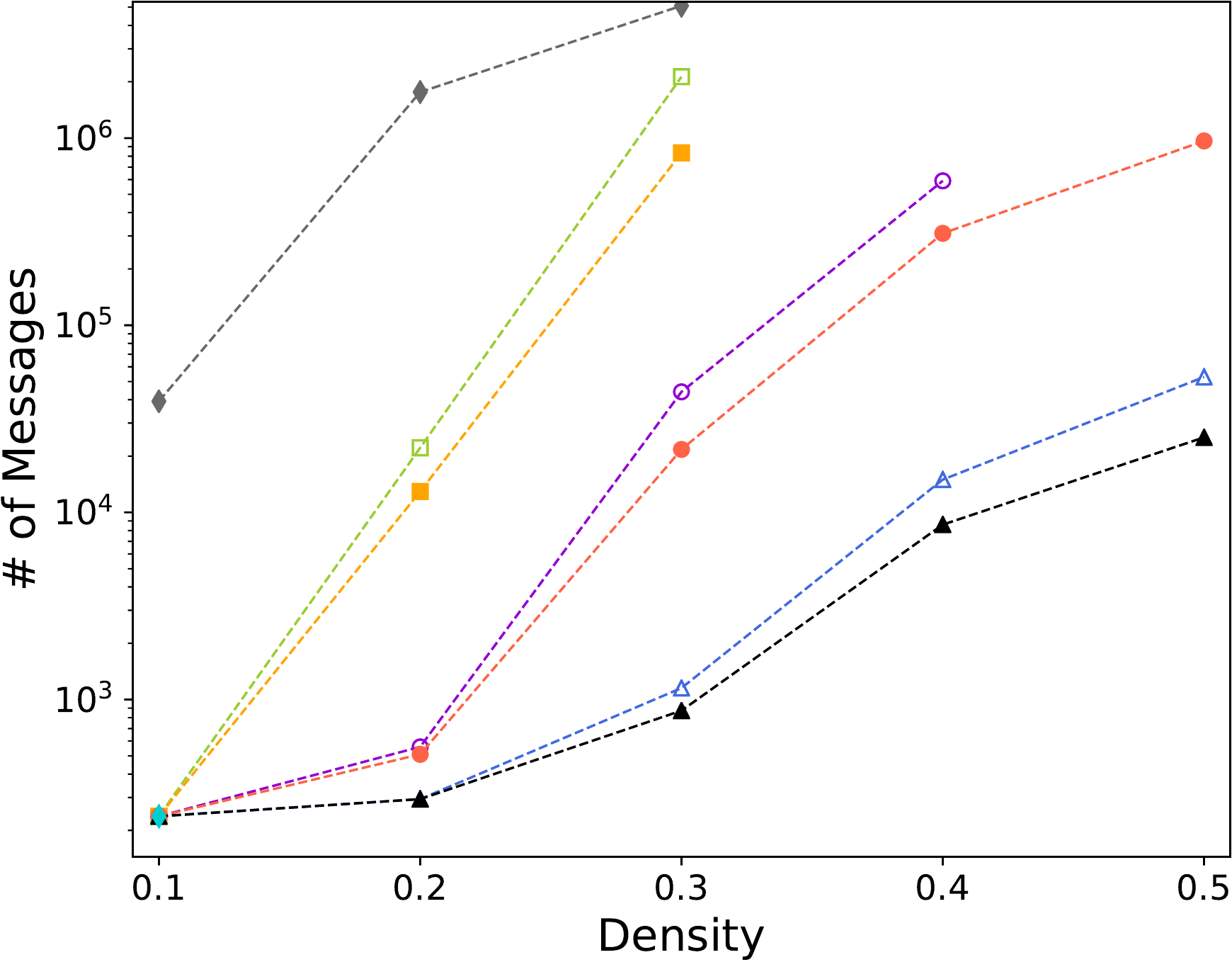} 
	}
	\subfloat[Network Load]{
		\includegraphics[scale=0.3]{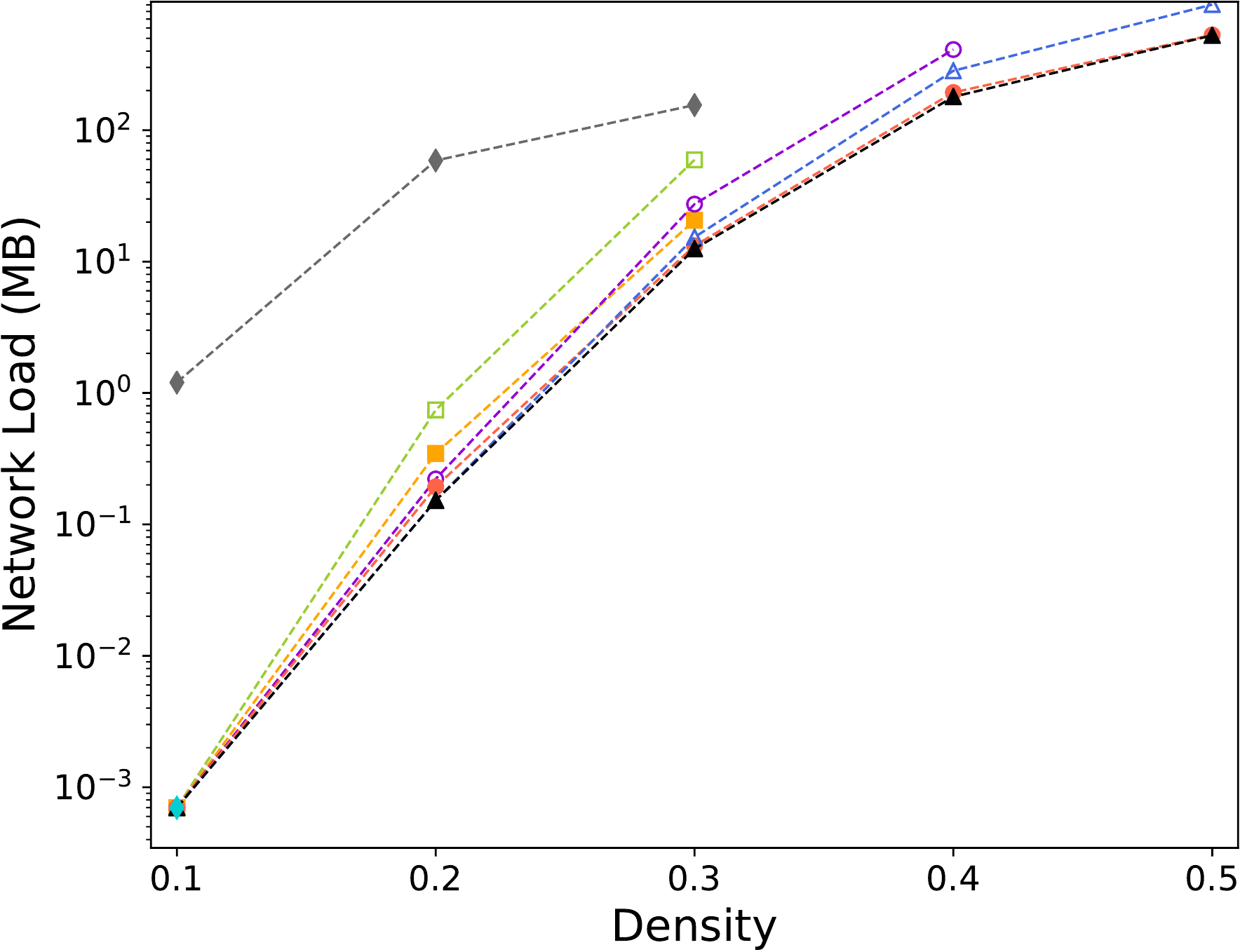} 
	}
	\subfloat[Runtime]{
		\includegraphics[scale=0.3]{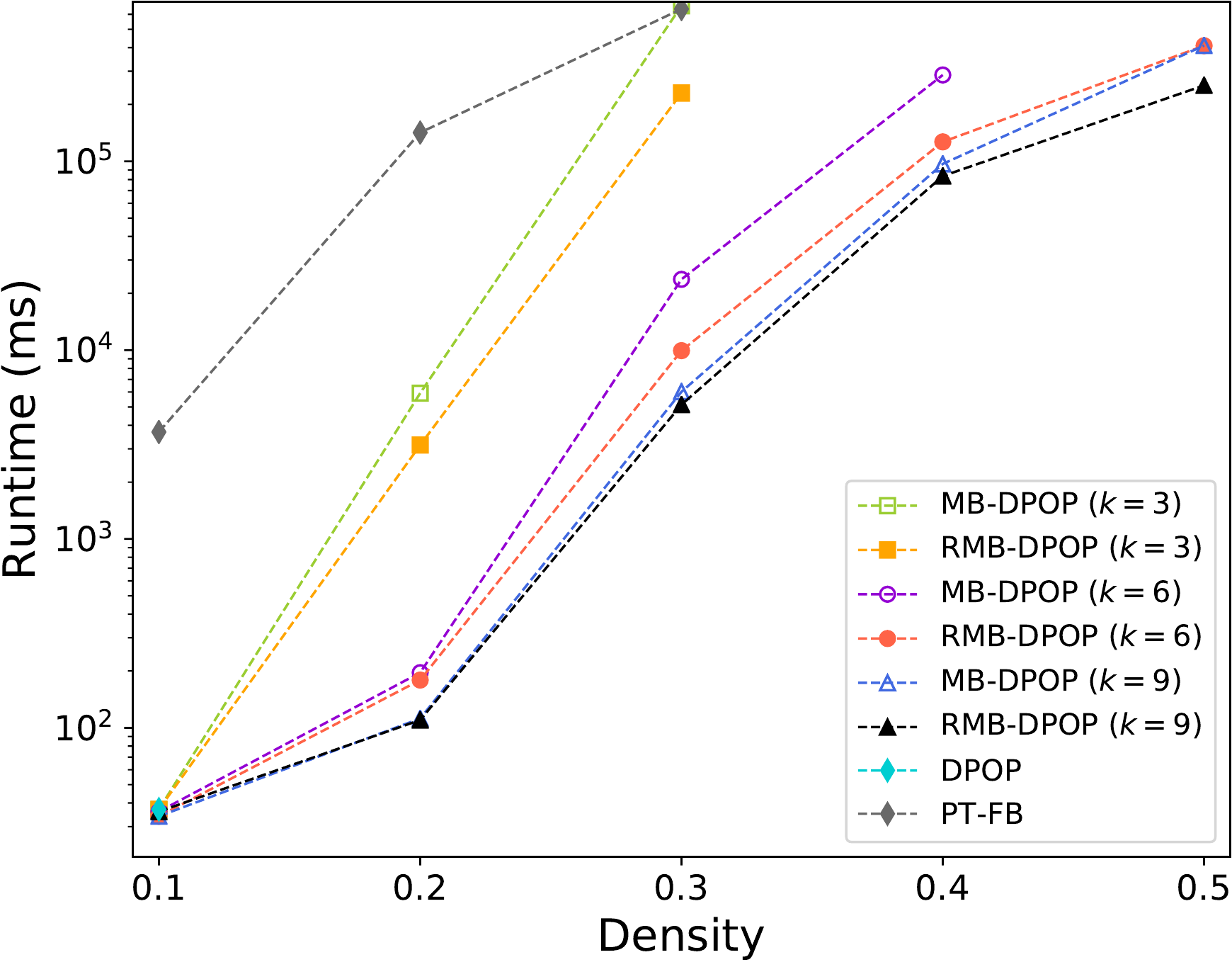} 
	}
	\caption{Performance comparison under different graph densities}
	\label{density}
\end{figure*}

\begin{table}
	\small
	\caption{The CClist of each node in the cluster}
	\begin{tabular}{|c|c|}
		\hline
		\textbf{$a_i$} & \textbf{$CClist_i$}       \\ \hline
		$a_3$          & $\{a_2,a_9,a_5,a_6,a_7\}$ \\ \hline
		$a_4$          & $\{a_2,a_5,a_6,a_7\}$     \\ \hline
		$a_5$          & $\{a_2,a_5,a_6,a_7\}$     \\ \hline
		$a_6$          & $\{a_2,a_5,a_6,a_7\}$     \\ \hline
		$a_7$          & $\{a_2,a_5,a_6,a_7\}$     \\ \hline
		$a_8$          & $\{a_2,a_5,a_6,a_7\}$     \\ \hline
		$a_9$          & $\{a_2,a_9\}$             \\ \hline
		$a_{10}$       & $\{a_2,a_9\}$             \\ \hline
		$a_{11}$       & $\{a_2,a_9\}$             \\ \hline
		$a_{12}$       & $\{a_2,a_9\}$             \\ \hline
		$a_{13}$       & $\{a_2,a_9\}$             \\ \hline
		$a_{14}$       & $\{a_2,a_9\}$             \\ \hline
	\end{tabular}
	\label{CClist}
\end{table}

It can be seen that node $a_2$ has the highest effectiveness and we should choose it as a CC node. Then a top-down phase is initiated to apply $a_i$ into the cluster. These two phases are performed alternatively until all the nodes in the cluster have a width less than $k$. The final CC nodes for each agent is listed as Table \ref{CClist}.

\begin{table}[]
	\small
	\caption{The first round of instantiation propagation}
	\begin{tabular}{|c|c|c|}
		\hline
		\textbf{step}      & \textbf{message}           & \textbf{INSTANTIATION}                 \\ \hline
		\multirow{2}{*}{1} & $a_3\rightarrow a_4$       & \{$a_2 = 0$\}                          \\ \cline{2-3} 
		& $a_3\rightarrow a_9$       & \{$a_2 = 0$\}                          \\ \hline
		\multirow{2}{*}{2} & $a_4\rightarrow a_5$       & \{$a_2 = 0$\}                          \\ \cline{2-3} 
		& $a_9\rightarrow a_{10}$    & \{$a_2 = 0, a_9 = 0$\}                 \\ \hline
		\multirow{3}{*}{3} & $a_5 \rightarrow a_6$      & \{$a_2 = 0, a_5 = 0$\}                 \\ \cline{2-3} 
		& $a_{10}\rightarrow a_{11}$ & \{$a_2 = 0, a_9 = 0$\}                 \\ \cline{2-3} 
		& $a_{10}\rightarrow a_{12}$ & \{$a_2 = 0, a_9 = 0$\}                 \\ \hline
		\multirow{2}{*}{4} & $a_6\rightarrow a_7$       & \{$a_2 = 0, a_5 = 0, a_6 = 0$\}        \\ \cline{2-3} 
		& $a_{12}\rightarrow a_{13}$ & \{$a_2 = 0, a_9 = 0$\}                 \\ \hline
		\multirow{2}{*}{5} & $a_7\rightarrow a_8$       & \{$a_2 = 0, a_5 = 0, a_6 = 0, a_7=0$\} \\ \cline{2-3} 
		& $a_{13}\rightarrow a_{14}$ & \{$a_2 = 0, a_9 = 0$\}                 \\ \hline
	\end{tabular}
	\label{instantiation}
\end{table}
Then, the DEM begins with $a_3$ which sends the first instantiation w.r.t. $CClist_3^{out}=\{a_2\}$ to its children $a_4$ and $a_9$. When receiving an instantiation, a CC node appends its assignment into the instantiation. Table \ref{instantiation} gives the trace of the first round of instantiation propagation.

\begin{figure}
\includegraphics[scale=0.3]{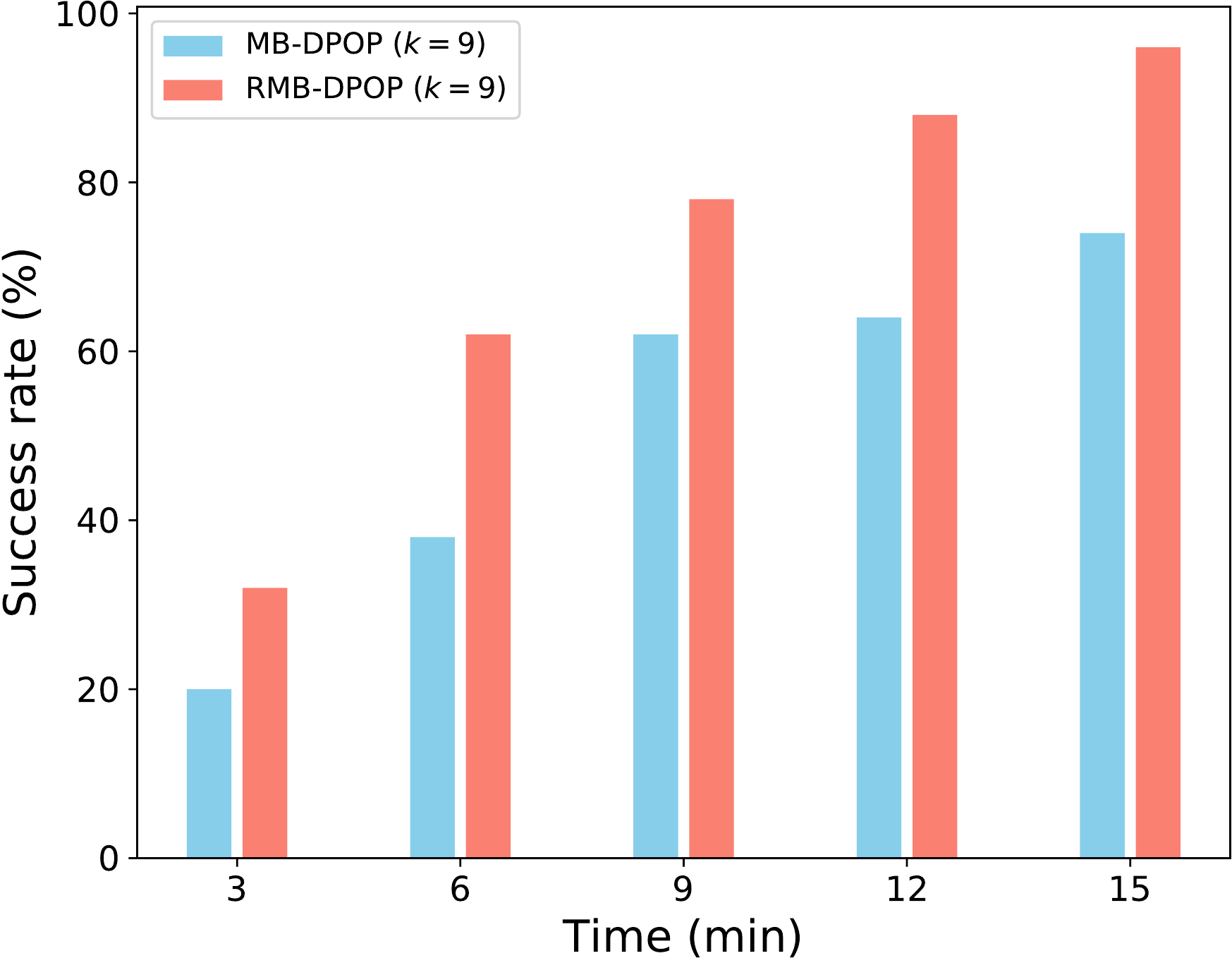}
\caption{Success rate within limited time}
\label{success}
\end{figure}

\begin{figure*}
	\subfloat[Number of Messages]{
		\includegraphics[scale=0.3]{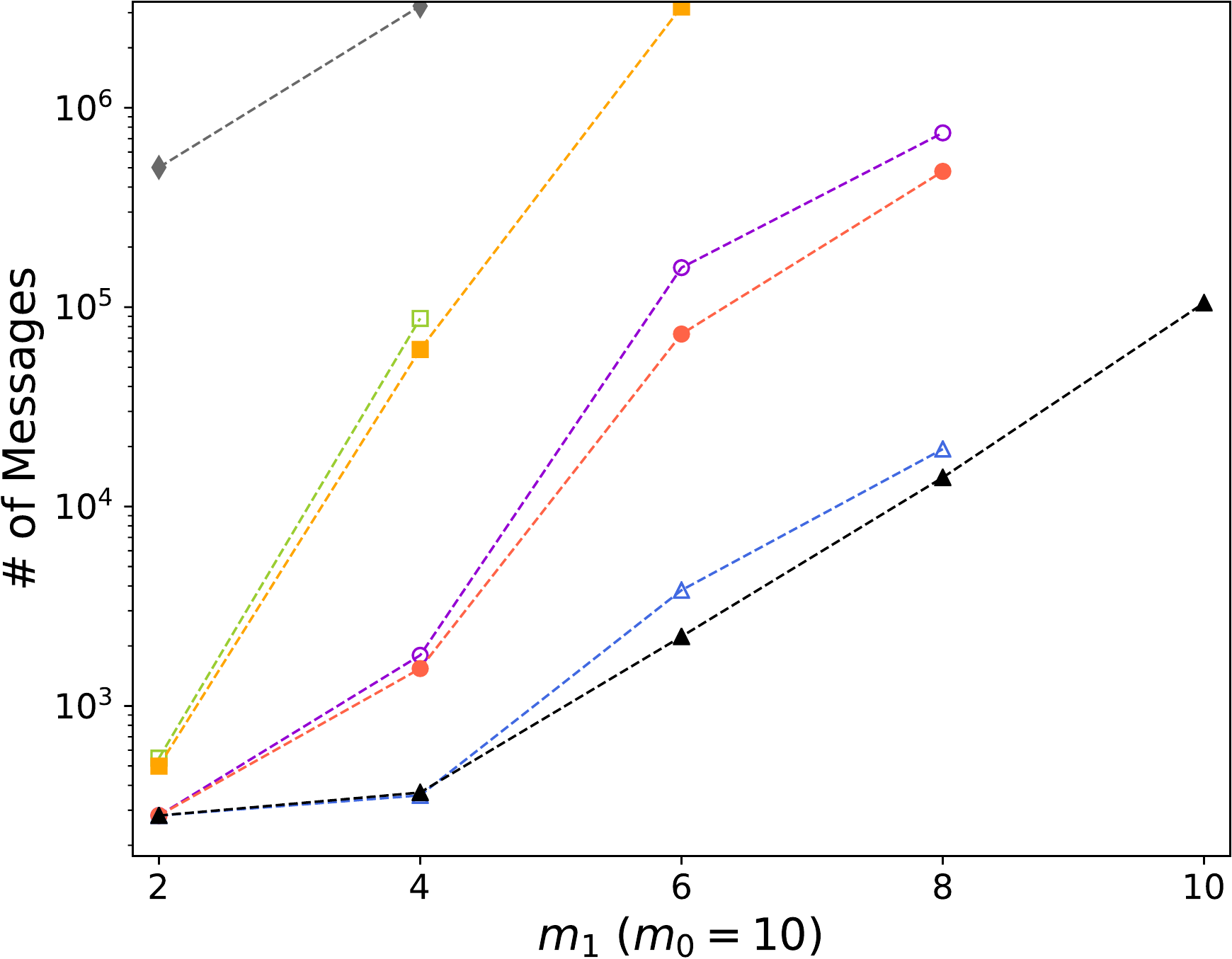} 
	}
	\subfloat[Network Load]{
		\includegraphics[scale=0.3]{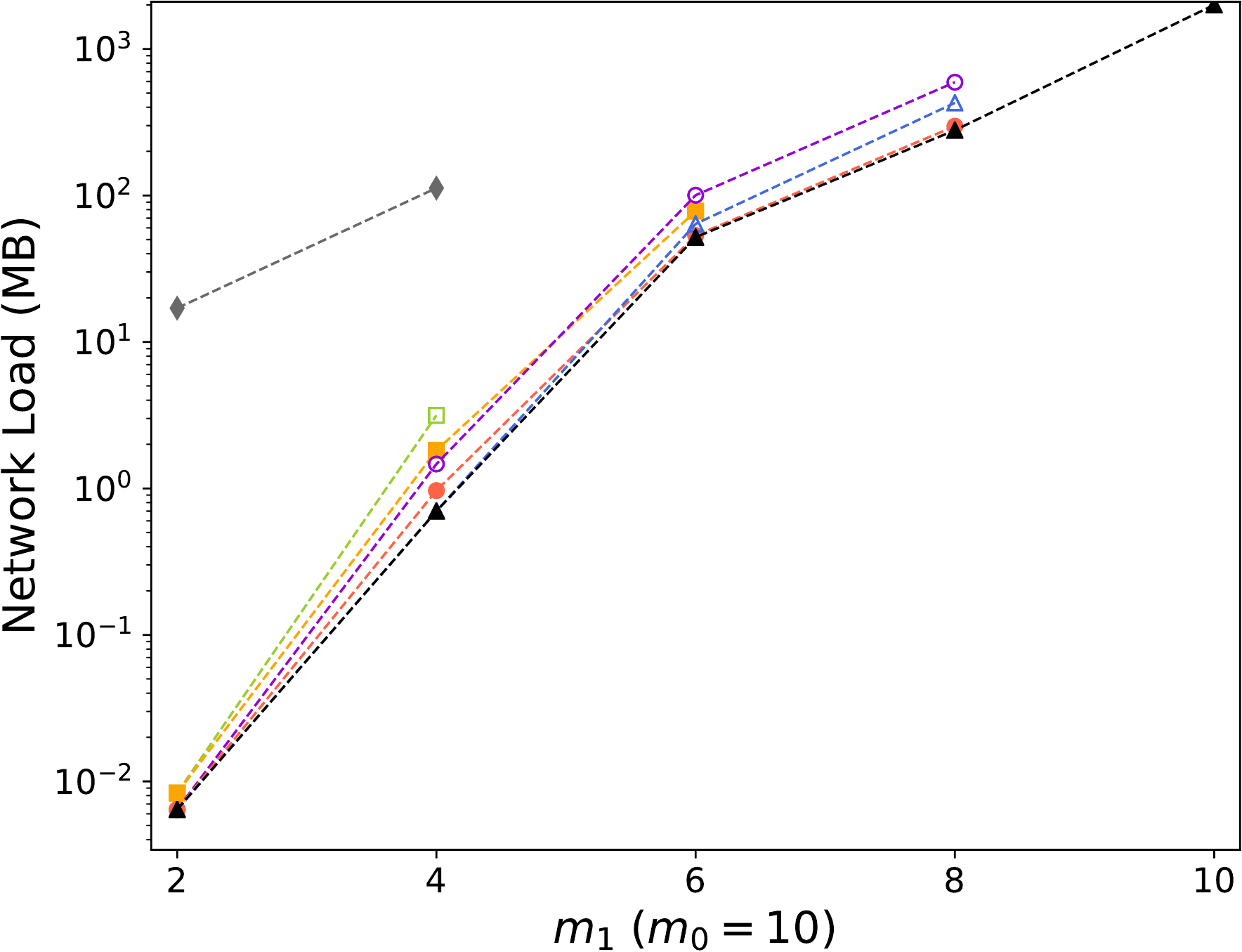} 
	}
	\subfloat[Runtime]{
		\includegraphics[scale=0.3]{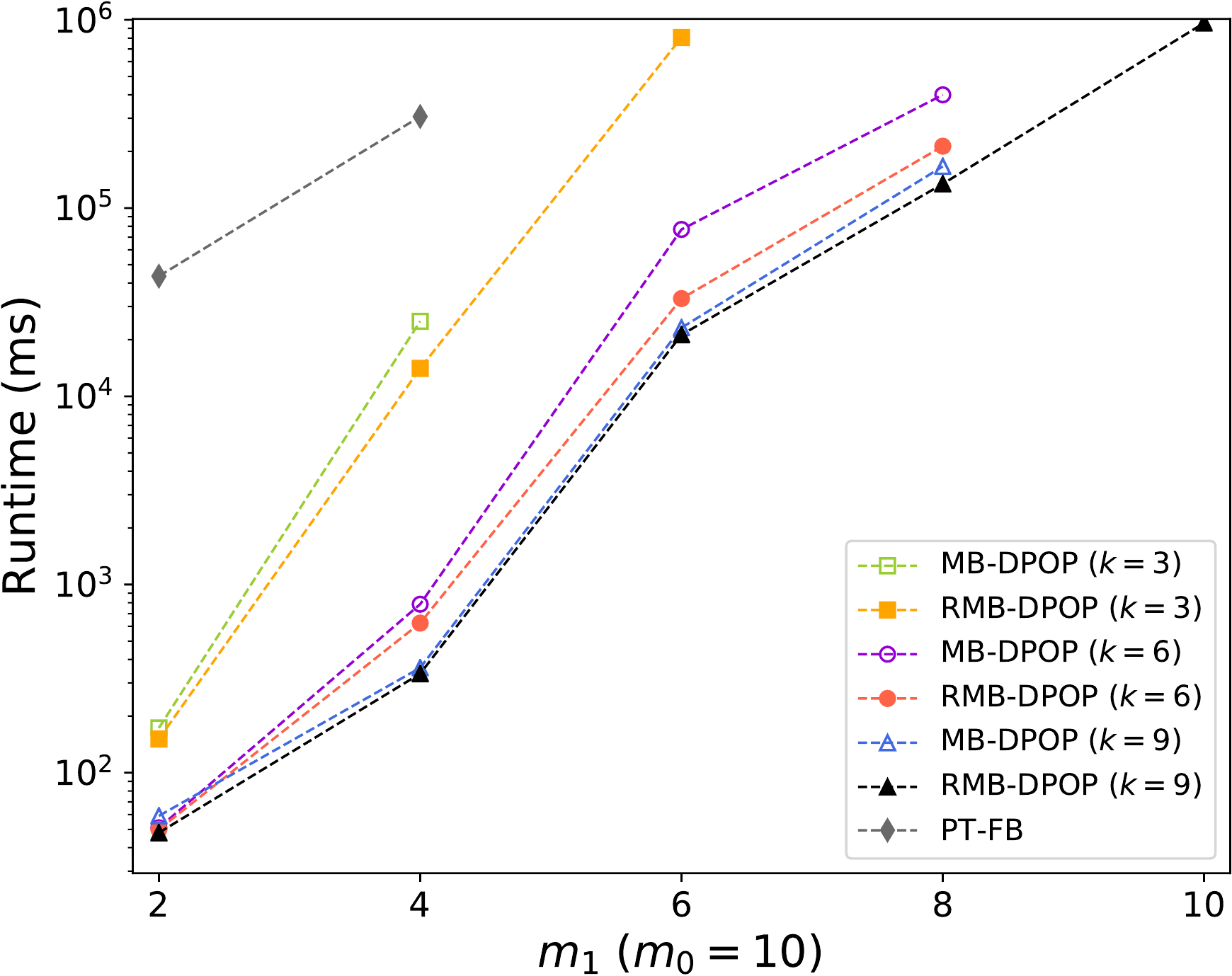} 
	}
	\caption{Performance comparison on scale-free networks}
	\label{scale-free}
\end{figure*}

\begin{figure}
	\includegraphics[scale=0.3]{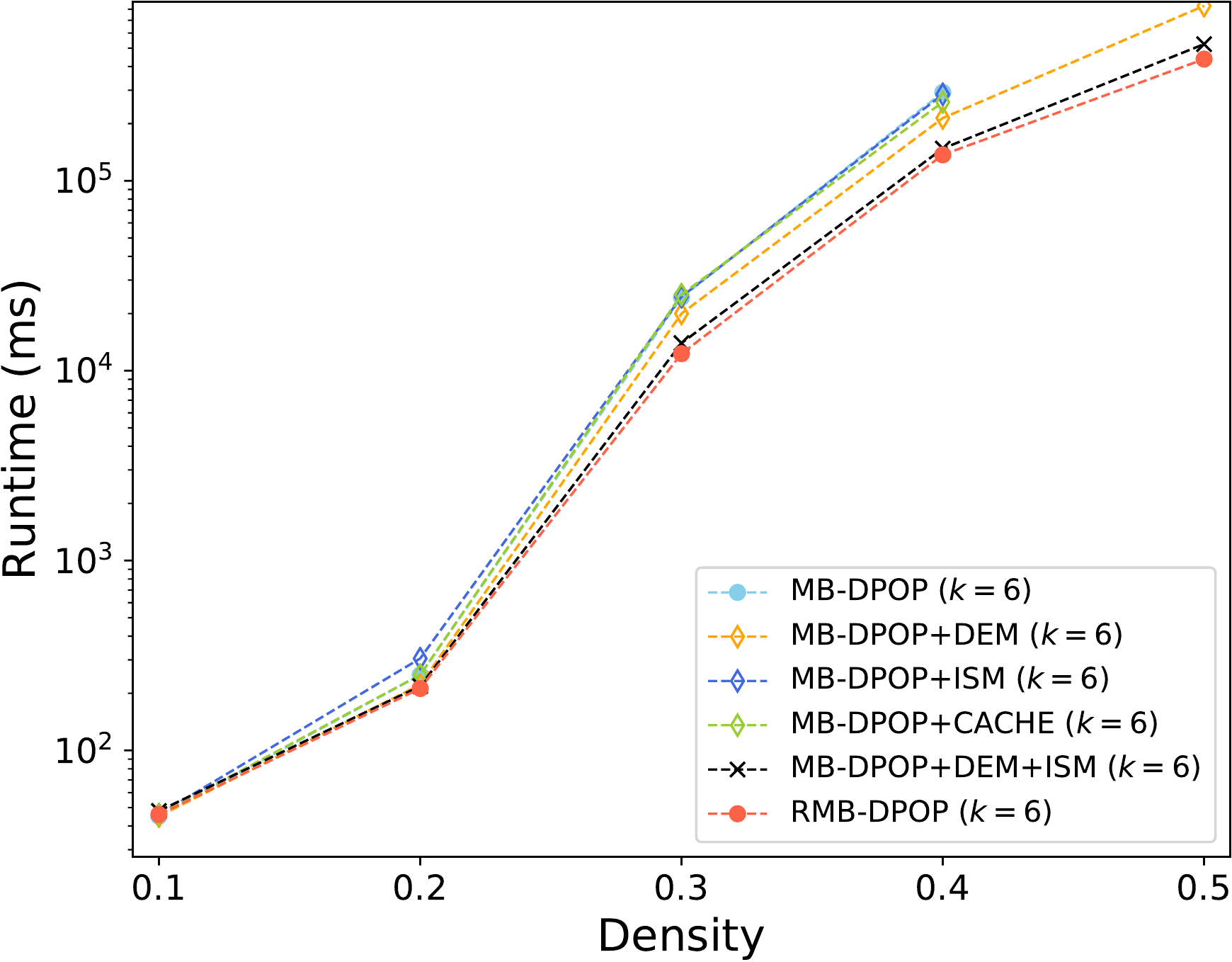}
	\caption{Performance of different mechanisms}
	\label{apart}
\end{figure}
 
\section{Experimental Evaluation}
In this section, we compare our proposed RMB-DPOP with the state-of-the-art on various benchmarks, and present an ablation study to demonstrate the effectiveness of each mechanism.

\subsection{Experimental Configuration}
We empirically evaluate RMB-DPOP, PT-FB, DPOP and MB-DPOP on two types of problems, i.e., random DCOPs and scale-free networks \cite{BarabEmergence}. In the first configuration, we consider the random DCOPs with the graph density of 0.2, the domain size of 3 and the agent number varying from 18 to 34. The second configuration is the DCOPs with 20 agents, the graph density of 0.2 and the domain size varying from 3 to 6. In addition, we present the ratio of the problems successfully solved within limited time on the second configuration where the graph density is set to 0.5. In the third configuration, we consider the scale-free networks generated by Barabási-Albert model where we set the agent number to 26, the domain size to 3 and $m_0$ to 10 and vary $m_1$ from 2 to 10.

In our experiments, we use the message number and network load (i.e., the size of total information exchanged) to measure the communication overheads. Also, we use wall clock time to measure the runtime. For each experiment, we generate 50 random instances and report the medium over all the instances. The experiments are conducted on an i7-7820x workstation with 32GB of memory and we set the timeout to 30 minutes for each algorithm. To demonstrate the effects of different $k$, we benchmark RMB-DPOP with different $k$ varying from 3 to 9. Finally, for fairness, we set the maximal number of dimensions for DPOP to 9.

\subsection{Experimental Results}
Fig. \ref{agent} presents the experiment results under different agent numbers. It can be seen from the figure that PT-FB cannot solve the problems with the agent number greater than 22. That is due to the fact that the search-based solvers need to explicitly exhaust the solution space by message-passing, which is quite expensive when solving the problems with the large agent number. Similarly, given the memory budget, DPOP also fails to scale up to the problems with the agent number more than 22. On the other hand, the scalability of the memory-bounded inference solvers depends on the size of $k$. For example, MB-DPOP ($k=3$) can only scale up to the problems with 26 agents while MB-DPOP ($k=9$) can solve the ones with 30 agents. This is because a large $k$ leads to fewer cycle-cut nodes and can significantly reduce the number of the memory-bounded inferences. Among these memory-bounded inference algorithms, given the same $k$, RMB-DPOP substantially outperforms MB-DPOP in both communication overheads and runtime. Besides, except for $k=6$, RMB-DPOP ($k$) can solve the problems with the larger agent number than MB-DPOP ($k$), which demonstrates our proposed mechanisms can improve the scalability of MB-DPOP. It is worth noting that the network load of RMB-DPOP ($k=3$) is less than MB-DPOP ($k=6$) when solving the problems with 30 agents, which indicates the merit of our proposed ISM. 

Fig. \ref{density} presents the results when solving the problems with different graph densities. It can be concluded from the figure that DPOP can only solve the problems with the density of 0.1 under this configuration. Besides, given the same $k$, RMB-DPOP ($k$) outperforms MB-DPOP ($k$) on all the metrics. Moreover, the gaps between RMB-DPOP and MB-DPOP are widen as the density grows, which demonstrates the potential of RMB-DPOP for reducing redundant inferences. Besides, it is noteworthy that RMB-DPOP with stricter memory budget can still outperform MB-DPOP with relatively large $k$ in terms of network load. For example, the network load of RMB-DPOP ($k=6$) is even less than the one of MB-DPOP ($k=9$) when solving dense problems. The same phenomenon also appears in solving the problems with the density of 0.3. The phenomena indicate that the redundant inference in MB-DPOP grows quickly as growing the graph density, while our proposed algorithm can effectively reduce unnecessary inferences. Fig. \ref{success} shows the ratio of the problems successfully solved within different time limits on this configuration where the graph density is 0.5. It can be seen that RMB-DPOP ($k=9$) solves over 90$\%$ of the problems in 15 minutes, while the success rate of MB-DPOP ($k=9$) is less than 80$\%$. Besides, RMB-DPOP ($k=9$) solves 60$\%$ problems in 6 minutes, while MB-DPOP ($k=9$) needs another 3 minutes (i.e., 9 minutes) to make that rate, which demonstrates the great superiority of the proposed algorithm again.

Fig. \ref{scale-free} shows the performance comparison when solving scale-free network problems with different $m_1$. PT-FB still cannot scale up due to prohibitively large search space. And it is worth mentioning that DPOP fails to solve all these problems, which demonstrates the poor scalability of DPOP under memory-limited scenarios. The reason is because the pseudo trees of scale-free network problems have induced width greater than 9 when $m_1=2$. On the other hand, our algorithm exhibits great advantage over MB-DPOP on all the metrics. The results show that RMB-DPOP ($k=9$) successfully solves all the problems, and except for $k=6$, RMB-DPOP ($k$) can solve the problems with larger $m_1$ than MB-DPOP ($k$). Although the scalability of RMB-DPOP ($k=6$) seem to be the same with the one of MB-DPOP ($k=6$), it can be seen that RMB-DPOP ($k=6$) incurs less network load than MB-DPOP ($k=9$) and its runtime closes to MB-DPOP ($k=9$) when $m_1 \ge 6$. It also demonstrates the scalability of our algorithm. 

Fig. \ref{apart} presents an ablation study on the second configuration with $k=6$ to demonstrate the effectiveness of each mechanism. It can be seen that the performance of MB-DPOP can be improved by each single mechanism when solving the dense problems, and can be further enhanced via the combinations of DEM and ISM. That is because ISM tends to choose the CC nodes inside clusters and DEM can effectively exploit these CC nodes to reduce the nonconcurrent instantiations. Without DEM and ISM, the contribution of caching mechanism is quite limited, but the combination of all the mechanisms achieves the best performance.

\section{Conclusions}
MB-DPOP suffers from a severe redundancy in memory-bounded inference due to the inability of exploiting the structure of a problem. In this paper, we propose a novel algorithm named RMB-DPOP which incorporates three mechanisms to reduce the redundancy in memory-bounded inference. First, we propose a distributed enumeration mechanism to make use of the independence among different branches to reduce the number of nonconcurrent instantiations. Second, we propose an iterative selection mechanism to refine the cycle-cut node selection, which aims to make each cycle-cut node to cover a maximum of the nodes with the width greater than $k$ in a cluster. Finally, a caching mechanism that exploits the historical inference results is introduced to further avoid unnecessary inferences. We theoretically prove that the distributed enumeration mechanism can reduce the message number if there is at least one cycle-cut node inside the clusters. Our experimental evaluations demonstrate the superiority of RMB-DPOP.

We note that our proposed mechanisms can be adapted to other algorithms as well. In more detail, the caching mechanism could be applied to an iterative process with recurrent combinations. Moreover, the selection mechanism could be used in other memory-bounded inference like ADPOP\cite{adpop2005} and HS-CAI\cite{hscai} to choose more appropriate variables to approximate or decimate. Also, this mechanism is highly customizable when combining with other algorithms. For example, we could easily implement different heuristics by changing the definition of $Eff_i$ for each candidate. Therefore, we envisage that these mechanisms not only advance the development of MB-DPOP, but also contribute to the algorithmic design of DCOPs.

\begin{acks}
We would like to thank the anonymous reviewers for their valuable comments and helpful suggestions. This work is partially supported by the Chongqing Research Program of Basic Research and Frontier Technology under Grant No.: cstc2017jcyjAX0030, the National Natural Science Foundation of China under Grant No.: 51608070 and the Graduate Research and Innovation Foundation of Chongqing, China under Grant No.: CYS18047.
\end{acks}


\bibliographystyle{ACM-Reference-Format}  
\bibliography{sample-bibliography}  


\begin{thebibliography}{00}


\ifx \showCODEN    \undefined \def \showCODEN     #1{\unskip}     \fi
\ifx \showDOI      \undefined \def \showDOI       #1{#1}\fi
\ifx \showISBNx    \undefined \def \showISBNx     #1{\unskip}     \fi
\ifx \showISBNxiii \undefined \def \showISBNxiii  #1{\unskip}     \fi
\ifx \showISSN     \undefined \def \showISSN      #1{\unskip}     \fi
\ifx \showLCCN     \undefined \def \showLCCN      #1{\unskip}     \fi
\ifx \shownote     \undefined \def \shownote      #1{#1}          \fi
\ifx \showarticletitle \undefined \def \showarticletitle #1{#1}   \fi
\ifx \showURL      \undefined \def \showURL       {\relax}        \fi
\providecommand\bibfield[2]{#2}
\providecommand\bibinfo[2]{#2}
\providecommand\natexlab[1]{#1}
\providecommand\showeprint[2][]{arXiv:#2}

\bibitem[\protect\citeauthoryear{Barabási and Albert}{Barabási and
  Albert}{1999}]%
        {BarabEmergence}
\bibfield{author}{\bibinfo{person}{Albert-László Barabási} {and}
  \bibinfo{person}{Réka Albert}.} \bibinfo{year}{1999}\natexlab{}.
\newblock \showarticletitle{Emergence of Scaling in Random Networks}.
\newblock \bibinfo{journal}{{\em Science\/}} \bibinfo{volume}{286},
  \bibinfo{number}{5439} (\bibinfo{year}{1999}), \bibinfo{pages}{509--512}.
\newblock


\bibitem[\protect\citeauthoryear{Brito and Meseguer}{Brito and
  Meseguer}{2010}]%
        {funcFilter2010}
\bibfield{author}{\bibinfo{person}{Ismel Brito} {and} \bibinfo{person}{Pedro
  Meseguer}.} \bibinfo{year}{2010}\natexlab{}.
\newblock \showarticletitle{Improving {DPOP} with function filtering}. In
  \bibinfo{booktitle}{{\em AAMAS}}. \bibinfo{pages}{141--148}.
\newblock


\bibitem[\protect\citeauthoryear{Chen, Deng, Chen, Zhang, and He}{Chen
  et~al\mbox{.}}{2019}]%
        {hscai}
\bibfield{author}{\bibinfo{person}{Dingding Chen}, \bibinfo{person}{Yanchen
  Deng}, \bibinfo{person}{Ziyu Chen}, \bibinfo{person}{Wenxin Zhang}, {and}
  \bibinfo{person}{Zhongshi He}.} \bibinfo{year}{2019}\natexlab{}.
\newblock \showarticletitle{{HS-CAI:} {A} Hybrid {DCOP} Algorithm via Combining
  Search with Context-based Inference}.
\newblock \bibinfo{journal}{{\em CoRR\/}}  \bibinfo{volume}{abs/1911.12716}
  (\bibinfo{year}{2019}).
\newblock
\showeprint[arxiv]{1911.12716}


\bibitem[\protect\citeauthoryear{Chen, Deng, Wu, and He}{Chen
  et~al\mbox{.}}{2018}]%
        {chen2018class}
\bibfield{author}{\bibinfo{person}{Ziyu Chen}, \bibinfo{person}{Yanchen Deng},
  \bibinfo{person}{Tengfei Wu}, {and} \bibinfo{person}{Zhongshi He}.}
  \bibinfo{year}{2018}\natexlab{}.
\newblock \showarticletitle{A class of iterative refined Max-sum algorithms via
  non-consecutive value propagation strategies}.
\newblock \bibinfo{journal}{{\em Autonomous Agents and Multi-Agent Systems\/}}
  \bibinfo{volume}{32}, \bibinfo{number}{6} (\bibinfo{year}{2018}),
  \bibinfo{pages}{822--860}.
\newblock


\bibitem[\protect\citeauthoryear{Dechter}{Dechter}{1999}]%
        {dechter1999bucket}
\bibfield{author}{\bibinfo{person}{Rina Dechter}.}
  \bibinfo{year}{1999}\natexlab{}.
\newblock \showarticletitle{Bucket elimination: A unifying framework for
  reasoning}.
\newblock \bibinfo{journal}{{\em Artificial Intelligence\/}}
  \bibinfo{volume}{113}, \bibinfo{number}{1-2} (\bibinfo{year}{1999}),
  \bibinfo{pages}{41--85}.
\newblock


\bibitem[\protect\citeauthoryear{Dechter, Cohen, et~al\mbox{.}}{Dechter
  et~al\mbox{.}}{2003}]%
        {cyclecut2003}
\bibfield{author}{\bibinfo{person}{Rina Dechter}, \bibinfo{person}{David
  Cohen}, {et~al\mbox{.}}} \bibinfo{year}{2003}\natexlab{}.
\newblock \bibinfo{booktitle}{{\em Constraint processing}}.
\newblock


\bibitem[\protect\citeauthoryear{Duan, Zhang, Mao, and Zhang}{Duan
  et~al\mbox{.}}{2018}]%
        {heterogeneous2018}
\bibfield{author}{\bibinfo{person}{Peibo Duan}, \bibinfo{person}{Changsheng
  Zhang}, \bibinfo{person}{Guoqiang Mao}, {and} \bibinfo{person}{Bin Zhang}.}
  \bibinfo{year}{2018}\natexlab{}.
\newblock \showarticletitle{Applying Distributed Constraint Optimization
  Approach to the User Association Problem in Heterogeneous Networks}.
\newblock \bibinfo{journal}{{\em {IEEE} Trans. Cybernetics\/}}
  \bibinfo{volume}{48}, \bibinfo{number}{6} (\bibinfo{year}{2018}),
  \bibinfo{pages}{1696--1707}.
\newblock


\bibitem[\protect\citeauthoryear{Farinelli, Rogers, Petcu, and
  Jennings}{Farinelli et~al\mbox{.}}{2008}]%
        {maxsum2008}
\bibfield{author}{\bibinfo{person}{Alessandro Farinelli}, \bibinfo{person}{Alex
  Rogers}, \bibinfo{person}{Adrian Petcu}, {and} \bibinfo{person}{Nicholas~R
  Jennings}.} \bibinfo{year}{2008}\natexlab{}.
\newblock \showarticletitle{Decentralised coordination of low-power embedded
  devices using the max-sum algorithm}. In \bibinfo{booktitle}{{\em AAMAS}}.
  \bibinfo{pages}{639--646}.
\newblock


\bibitem[\protect\citeauthoryear{Fioretto, Yeoh, Pontelli, Ma, and
  Ranade}{Fioretto et~al\mbox{.}}{2017}]%
        {grid2017}
\bibfield{author}{\bibinfo{person}{Ferdinando Fioretto},
  \bibinfo{person}{William Yeoh}, \bibinfo{person}{Enrico Pontelli},
  \bibinfo{person}{Ye Ma}, {and} \bibinfo{person}{Satishkumar~J Ranade}.}
  \bibinfo{year}{2017}\natexlab{}.
\newblock \showarticletitle{A distributed constraint optimization ({DCOP})
  approach to the economic dispatch with demand response}. In
  \bibinfo{booktitle}{{\em AAMAS}}. \bibinfo{pages}{999--1007}.
\newblock


\bibitem[\protect\citeauthoryear{Freuder and Quinn}{Freuder and Quinn}{1985}]%
        {freuderQ85}
\bibfield{author}{\bibinfo{person}{Eugene~C. Freuder} {and}
  \bibinfo{person}{Michael~J. Quinn}.} \bibinfo{year}{1985}\natexlab{}.
\newblock \showarticletitle{Taking Advantage of Stable Sets of Variables in
  Constraint Satisfaction Problems}. In \bibinfo{booktitle}{{\em IJCAI}}.
  \bibinfo{pages}{1076--1078}.
\newblock


\bibitem[\protect\citeauthoryear{Gershman, Meisels, and Zivan}{Gershman
  et~al\mbox{.}}{2009}]%
        {afb2009}
\bibfield{author}{\bibinfo{person}{Amir Gershman}, \bibinfo{person}{Amnon
  Meisels}, {and} \bibinfo{person}{Roie Zivan}.}
  \bibinfo{year}{2009}\natexlab{}.
\newblock \showarticletitle{Asynchronous forward bounding for distributed
  COPs}.
\newblock \bibinfo{journal}{{\em Journal of Artificial Intelligence
  Research\/}}  \bibinfo{volume}{34} (\bibinfo{year}{2009}),
  \bibinfo{pages}{61--88}.
\newblock


\bibitem[\protect\citeauthoryear{Gutierrez and Meseguer}{Gutierrez and
  Meseguer}{2012}]%
        {bnbadopt_remove2012}
\bibfield{author}{\bibinfo{person}{Patricia Gutierrez} {and}
  \bibinfo{person}{Pedro Meseguer}.} \bibinfo{year}{2012}\natexlab{}.
\newblock \showarticletitle{Removing redundant messages in n-ary BnB-ADOPT}.
\newblock \bibinfo{journal}{{\em Journal of Artificial Intelligence
  Research\/}}  \bibinfo{volume}{45} (\bibinfo{year}{2012}),
  \bibinfo{pages}{287--304}.
\newblock


\bibitem[\protect\citeauthoryear{Gutierrez, Meseguer, and Yeoh}{Gutierrez
  et~al\mbox{.}}{2011}]%
        {bnbadopt_generalizing2011}
\bibfield{author}{\bibinfo{person}{Patricia Gutierrez}, \bibinfo{person}{Pedro
  Meseguer}, {and} \bibinfo{person}{William Yeoh}.}
  \bibinfo{year}{2011}\natexlab{}.
\newblock \showarticletitle{Generalizing adopt and bnb-adopt}. In
  \bibinfo{booktitle}{{\em IJCAI}}. \bibinfo{pages}{554--559}.
\newblock


\bibitem[\protect\citeauthoryear{Hirayama and Yokoo}{Hirayama and
  Yokoo}{1997}]%
        {sbb1997}
\bibfield{author}{\bibinfo{person}{Katsutoshi Hirayama} {and}
  \bibinfo{person}{Makoto Yokoo}.} \bibinfo{year}{1997}\natexlab{}.
\newblock \showarticletitle{Distributed partial constraint satisfaction
  problem}. In \bibinfo{booktitle}{{\em CP}}. \bibinfo{pages}{222--236}.
\newblock


\bibitem[\protect\citeauthoryear{Li, Negenborn, and Lodewijks}{Li
  et~al\mbox{.}}{2016}]%
        {vessel2016}
\bibfield{author}{\bibinfo{person}{Shijie Li}, \bibinfo{person}{Rudy~R.
  Negenborn}, {and} \bibinfo{person}{Gabri{\"{e}}l Lodewijks}.}
  \bibinfo{year}{2016}\natexlab{}.
\newblock \showarticletitle{Distributed constraint optimization for addressing
  vessel rotation planning problems}.
\newblock \bibinfo{journal}{{\em Engineering Applications of Artificial
  Intelligence\/}}  \bibinfo{volume}{48} (\bibinfo{year}{2016}),
  \bibinfo{pages}{159--172}.
\newblock


\bibitem[\protect\citeauthoryear{Litov and Meisels}{Litov and Meisels}{2017}]%
        {ptfb2017}
\bibfield{author}{\bibinfo{person}{Omer Litov} {and} \bibinfo{person}{Amnon
  Meisels}.} \bibinfo{year}{2017}\natexlab{}.
\newblock \showarticletitle{Forward bounding on pseudo-trees for DCOPs and
  ADCOPs}.
\newblock \bibinfo{journal}{{\em Artificial Intelligence\/}}
  \bibinfo{volume}{252} (\bibinfo{year}{2017}), \bibinfo{pages}{83--99}.
\newblock


\bibitem[\protect\citeauthoryear{Maheswaran, Pearce, and Tambe}{Maheswaran
  et~al\mbox{.}}{2004}]%
        {mgm2004}
\bibfield{author}{\bibinfo{person}{Rajiv~T Maheswaran},
  \bibinfo{person}{Jonathan~P Pearce}, {and} \bibinfo{person}{Milind Tambe}.}
  \bibinfo{year}{2004}\natexlab{}.
\newblock \showarticletitle{Distributed Algorithms for {DCOP}: {A}
  Graphical-Game-Based Approach.}. In \bibinfo{booktitle}{{\em ISCA PDCS}}.
  \bibinfo{pages}{432--439}.
\newblock


\bibitem[\protect\citeauthoryear{Modi, Shen, Tambe, and Yokoo}{Modi
  et~al\mbox{.}}{2005}]%
        {adopt2005}
\bibfield{author}{\bibinfo{person}{Pragnesh~Jay Modi}, \bibinfo{person}{Wei-Min
  Shen}, \bibinfo{person}{Milind Tambe}, {and} \bibinfo{person}{Makoto Yokoo}.}
  \bibinfo{year}{2005}\natexlab{}.
\newblock \showarticletitle{ADOPT: Asynchronous distributed constraint
  optimization with quality guarantees}.
\newblock \bibinfo{journal}{{\em Artificial Intelligence\/}}
  \bibinfo{volume}{161}, \bibinfo{number}{1-2} (\bibinfo{year}{2005}),
  \bibinfo{pages}{149--180}.
\newblock


\bibitem[\protect\citeauthoryear{Nguyen, Yeoh, Lau, and Zivan}{Nguyen
  et~al\mbox{.}}{2019}]%
        {dgibbs2019}
\bibfield{author}{\bibinfo{person}{Duc~Thien Nguyen}, \bibinfo{person}{William
  Yeoh}, \bibinfo{person}{Hoong~Chuin Lau}, {and} \bibinfo{person}{Roie
  Zivan}.} \bibinfo{year}{2019}\natexlab{}.
\newblock \showarticletitle{Distributed Gibbs: {A} Linear-Space Sampling-Based
  {DCOP} Algorithm}.
\newblock \bibinfo{journal}{{\em Journal of Artificial Intelligence
  Research\/}}  \bibinfo{volume}{64} (\bibinfo{year}{2019}),
  \bibinfo{pages}{705--748}.
\newblock


\bibitem[\protect\citeauthoryear{Pearce and Tambe}{Pearce and Tambe}{2007}]%
        {k-optimal2007}
\bibfield{author}{\bibinfo{person}{Jonathan~P Pearce} {and}
  \bibinfo{person}{Milind Tambe}.} \bibinfo{year}{2007}\natexlab{}.
\newblock \showarticletitle{Quality Guarantees on k-Optimal Solutions for
  Distributed Constraint Optimization Problems.}. In \bibinfo{booktitle}{{\em
  IJCAI}}. \bibinfo{pages}{1446--1451}.
\newblock


\bibitem[\protect\citeauthoryear{Petcu and Faltings}{Petcu and
  Faltings}{2005a}]%
        {adpop2005}
\bibfield{author}{\bibinfo{person}{Adrian Petcu} {and} \bibinfo{person}{Boi
  Faltings}.} \bibinfo{year}{2005}\natexlab{a}.
\newblock \showarticletitle{Approximations in distributed optimization}. In
  \bibinfo{booktitle}{{\em CP}}. \bibinfo{pages}{802--806}.
\newblock


\bibitem[\protect\citeauthoryear{Petcu and Faltings}{Petcu and
  Faltings}{2005b}]%
        {dpop2005}
\bibfield{author}{\bibinfo{person}{Adrian Petcu} {and} \bibinfo{person}{Boi
  Faltings}.} \bibinfo{year}{2005}\natexlab{b}.
\newblock \showarticletitle{A scalable method for multiagent constraint
  optimization}. In \bibinfo{booktitle}{{\em IJCAI}}.
  \bibinfo{pages}{266--271}.
\newblock


\bibitem[\protect\citeauthoryear{Petcu and Faltings}{Petcu and
  Faltings}{2006}]%
        {odpop2006}
\bibfield{author}{\bibinfo{person}{Adrian Petcu} {and} \bibinfo{person}{Boi
  Faltings}.} \bibinfo{year}{2006}\natexlab{}.
\newblock \showarticletitle{{ODPOP}: {A}n algorithm for open/distributed
  constraint optimization}. In \bibinfo{booktitle}{{\em AAAI}}.
  \bibinfo{pages}{703--708}.
\newblock


\bibitem[\protect\citeauthoryear{Petcu and Faltings}{Petcu and
  Faltings}{2007a}]%
        {lsdpop2007}
\bibfield{author}{\bibinfo{person}{Adrian Petcu} {and} \bibinfo{person}{Boi
  Faltings}.} \bibinfo{year}{2007}\natexlab{a}.
\newblock \showarticletitle{A hybrid of inference and local search for
  distributed combinatorial optimization}. In \bibinfo{booktitle}{{\em ICIAT}}.
  \bibinfo{pages}{342--348}.
\newblock


\bibitem[\protect\citeauthoryear{Petcu and Faltings}{Petcu and
  Faltings}{2007b}]%
        {mbdpop2007}
\bibfield{author}{\bibinfo{person}{Adrian Petcu} {and} \bibinfo{person}{Boi
  Faltings}.} \bibinfo{year}{2007}\natexlab{b}.
\newblock \showarticletitle{{MB-DPOP}: {A} New Memory-Bounded Algorithm for
  Distributed Optimization}. In \bibinfo{booktitle}{{\em IJCAI}}.
  \bibinfo{pages}{1452--1457}.
\newblock


\bibitem[\protect\citeauthoryear{Sultanik, Modi, and Regli}{Sultanik
  et~al\mbox{.}}{2007}]%
        {taskmodeling2007}
\bibfield{author}{\bibinfo{person}{Evan Sultanik},
  \bibinfo{person}{Pragnesh~Jay Modi}, {and} \bibinfo{person}{William~C
  Regli}.} \bibinfo{year}{2007}\natexlab{}.
\newblock \showarticletitle{On Modeling Multiagent Task Scheduling as a
  Distributed Constraint Optimization Problem.}. In \bibinfo{booktitle}{{\em
  IJCAI}}. \bibinfo{pages}{1531--1536}.
\newblock


\bibitem[\protect\citeauthoryear{Yeoh, Felner, and Koenig}{Yeoh
  et~al\mbox{.}}{2010}]%
        {bnbadopt2010}
\bibfield{author}{\bibinfo{person}{William Yeoh}, \bibinfo{person}{Ariel
  Felner}, {and} \bibinfo{person}{Sven Koenig}.}
  \bibinfo{year}{2010}\natexlab{}.
\newblock \showarticletitle{BnB-ADOPT: An asynchronous branch-and-bound DCOP
  algorithm}.
\newblock \bibinfo{journal}{{\em Journal of Artificial Intelligence
  Research\/}}  \bibinfo{volume}{38} (\bibinfo{year}{2010}),
  \bibinfo{pages}{85--133}.
\newblock


\bibitem[\protect\citeauthoryear{Yeoh, Sun, and Koenig}{Yeoh
  et~al\mbox{.}}{2009a}]%
        {adopt_trading2009}
\bibfield{author}{\bibinfo{person}{William Yeoh}, \bibinfo{person}{Xiaoxun
  Sun}, {and} \bibinfo{person}{Sven Koenig}.} \bibinfo{year}{2009}\natexlab{a}.
\newblock \showarticletitle{Trading off solution quality for faster computation
  in DCOP search algorithms}. In \bibinfo{booktitle}{{\em IJCAI}}.
  \bibinfo{pages}{354--360}.
\newblock


\bibitem[\protect\citeauthoryear{Yeoh, Varakantham, and Koenig}{Yeoh
  et~al\mbox{.}}{2009b}]%
        {adopt_caching2009}
\bibfield{author}{\bibinfo{person}{William Yeoh}, \bibinfo{person}{Pradeep
  Varakantham}, {and} \bibinfo{person}{Sven Koenig}.}
  \bibinfo{year}{2009}\natexlab{b}.
\newblock \showarticletitle{Caching schemes for DCOP search algorithms}. In
  \bibinfo{booktitle}{{\em AAMAS}}. \bibinfo{pages}{609--616}.
\newblock


\bibitem[\protect\citeauthoryear{Zhang, Wang, Xing, and Wittenburg}{Zhang
  et~al\mbox{.}}{2005}]%
        {dsa2005}
\bibfield{author}{\bibinfo{person}{Weixiong Zhang}, \bibinfo{person}{Guandong
  Wang}, \bibinfo{person}{Zhao Xing}, {and} \bibinfo{person}{Lars Wittenburg}.}
  \bibinfo{year}{2005}\natexlab{}.
\newblock \showarticletitle{Distributed stochastic search and distributed
  breakout: {P}roperties, comparison and applications to constraint
  optimization problems in sensor networks}.
\newblock \bibinfo{journal}{{\em Artificial Intelligence\/}}
  \bibinfo{volume}{161}, \bibinfo{number}{1-2} (\bibinfo{year}{2005}),
  \bibinfo{pages}{55--87}.
\newblock


\end{thebibliography}

\end{document}